\documentclass[11pt]{article}
\usepackage[papersize={8.5in,11in},top=1.1in,left=1.1in,right=1.1in,bottom=1.1in]{geometry}
\usepackage{cite,amsmath,amsthm,amsfonts}
\usepackage{pstricks,pst-node}

\newtheorem{theorem}{Theorem}[section]
\newtheorem{cor}[theorem]{Corollary}
\newtheorem{lemma}[theorem]{Lemma}

\newtheorem{claim}[theorem]{Claim}
\newtheorem{example}[theorem]{Example}

\newtheorem{definition}[theorem]{Definition}

\newcommand{\classNP}{{\sf NP}}
\newcommand{\classP}{{\sf P}}

\newcommand{\R}{{\mathbb{R}}}
\newcommand{\e}{{\varepsilon}}

\begin{document}

\title{Contribution Games in Networks
  \thanks{An extended abstract of this paper has been accepted for
    publication in the proceedings of the 18th European Symposium on
    Algorithms (ESA 2010). }
}

\author{
Elliot Anshelevich
\thanks{Dept. of Computer Science, Rensselaer Polytechnic Institute, Troy, NY. Supported in part by NSF CCF-0914782.}
\and
Martin Hoefer
\thanks{Dept. of Computer Science, RWTH Aachen University, Germany. {\tt mhoefer@cs.rwth-aachen.de}. Supported by DFG through UMIC Research Center at RWTH Aachen University and by grant Ho 3831/3-1.}
}
\date{}

\maketitle

\begin{abstract}
  We consider {\em network contribution games}, where each agent in a
  network has a budget of effort that he can contribute to different
  collaborative projects or relationships. Depending on the
  contribution of the involved agents a relationship will flourish or
  drown, and to measure the success we use a reward function for each
  relationship. Every agent is trying to maximize the reward from all
  relationships that it is involved in. We consider pairwise
  equilibria of this game, and characterize the existence,
  computational complexity, and quality of equilibrium based on the
  types of reward functions involved. When all reward functions are
  concave, we prove that
  the price of anarchy is at most 2. For convex functions the same
  only holds under some special but very natural conditions. Another
  special case extensively treated are minimum effort games, where the
  reward of a relationship depends only on the minimum effort of any
  of the participants. In these games, we can show existence of
  pairwise equilibrium and a price of anarchy of 2 for concave
  functions and special classes of games with convex
  functions. Finally, we show tight bounds for approximate equilibria
  and convergence of dynamics in these games.
\end{abstract}

\section{Introduction}
\label{sec:intro}

Understanding the degree to which rational agents will participate in
and contribute to joint projects is critical in many areas of
society. With the advent of the Internet and the consideration of
rationality in the design of multi-agent and peer-to-peer systems,
these aspects are becoming of interest to computer scientists and
subject to analytical computer science research. Not surprisingly, the
study of contribution incentives has been an area of vital research
interest in economics and related areas with seminal contributions to
the topic over the last decades. A prominent example from experimental
economics is the \emph{minimum effort coordination
  game}~\cite{Huyck90}, in which a number of participants contribute
to a joint project, and the outcome depends solely on the minimum
contribution of any agent. While the Nash equilibria in this game
exhibit a quite simple structure, behavior in laboratory experiments
led to sometimes surprising patterns see,
e.g.,~\cite{Dufwenberg05,Bornstein02,Fatas06,Riechmann08,Chaudhuri08}
for recent examples. On the analytical side this game was studied, for
instance, with respect to logit-response dynamics and stochastic
potential in~\cite{Anderson01}.
%
%

In this paper we propose and study a simple framework of \emph{network
  contribution games} for contribution, collaboration, and
coordination of actors embedded in networks. The game contains the
minimum effort coordination game as a special case and is closely
related to many other games from the economics literature. In such a
game each player is a vertex in a graph, and the edges represent
bilateral relationships that he can engage in. Each player has a
budget of effort that he can contribute to different edges. Budgets
and contributions are non-negative numbers, and we use them as an
abstraction for the different ways and degrees by which actors can
contribute to a relationship, e.g., by allocating time, money, and
personal energy to maintaining a friendship or a collaboration.
Depending on the contribution of the involved actors a relationship
will flourish or drown, and to measure the success we use a reward
function for each relationship. Finally, each player strives to
maximize the total success of all relationships he is involved in.


A major issue that we address in our games is the impact of
collaboration.  An incentive for collaboration evolves naturally when
agents are embedded in (social) networks and engage in
relationships. We are interested in the way that a limited
collaboration between agents influences properties of equilibria in
contribution games like existence, computational complexity, the
convergence of natural dynamics, as well as measures of
inefficiency. In particular, in addition to unilateral strategy
changes we will allow pairs of players to change their strategies in a
coordinated manner. States that are resilient against such bilateral
deviations are termed \emph{2-strong}~\cite{Andelman09} or
\emph{pairwise equilibria}~\cite{Jackson96}. This adjustment raises a
number of interesting questions. What is the structure of pairwise
equilibria, and what are conditions under which they exist?  Can we
compute pairwise equilibria efficiently or at least efficiently decide
their existence? Are there natural improvement dynamics that allow
players to reach a pairwise equilibrium (quickly)?  What are the
prices of anarchy and stability, i.e., the ratios of the social
welfare of the best possible state over the worst and best welfare of
an equilibrium, respectively?  These are the main questions that
motivate our study.
Before describing our results, we proceed with a formal introduction of
the model.

\subsection{Network Contribution Games}
\label{sec:model}
We consider \emph{network contribution games} as models for the
contribution to relationships in networked environments. In our games
we are given a simple and undirected graph $G=(V,E)$ with $n$ nodes
and $m$ edges. Every node $v \in V$ is a \emph{player}, and every edge
$e \in E$ represents a \emph{relationship} (collaboration, friendship,
etc.). A player $v$ has a given \emph{budget} $B_v\geq 0$ of the total
amount of \emph{effort} that it is able to apply to all of its
relationships (i.e., edges incident to $v$). Budgets are called
\emph{uniform} if $B_u = B_v$ for any $u,v \in V$. In this case,
unless stated otherwise, we assume that $B_v = 1$ for all $v \in V$
and scale reward functions accordingly.

We denote by $E_v$ the set of edges incident to $v$. A \emph{strategy}
for player $v$ is a function $s_v : E_v \to \R_{\ge 0}$ that satisfies
$\sum_{e=(v,u)} s_v(e) \le B_v$ and specifies the amount of effort
$s_v(e)$ that $v$ puts into relationship $e \in E_v$. A \emph{state}
of the game is simply a vector $s = (s_1,\ldots,s_n)$. The success of
a relationship $e$ is measured by a \emph{reward function} $f_e :
\R_{\ge 0}^2 \rightarrow \R$, for which $f_e(x,y) \ge 0$ and
non-decreasing in $x,y \ge 0$. The \emph{utility} or \emph{welfare} of
a player $v$ is simply the total success of all its relationships,
i.e., $w_v(s) = \sum_{e=(v,u)}f_e(s_v(e),s_u(e))$, so both endpoints
benefit equally from the undirected edge $e$. In addition, we will
assume that reward functions $f_e$ are symmetric, so
$f_e(x,y)=f_e(y,x)$ for all $x,y \ge 0$, and for ease of presentation
we will assume they are continuous and differentiable, although most
of our results can be obtained without these assumptions.

We are interested in the existence and computational complexity of
stable states, their performance in terms of social welfare, and the
convergence of natural dynamics to equilibrium. The central concept of
stability in strategic games is the \emph{(pure) Nash equilibrium},
which is resilient against unilateral deviations, i.e., a state $s$
such that $w_v(s_v, s_{-v}) \ge w_v(s_v',s_{-v})$ for each $v \in V$
and all possible strategies $s_v'$. For the \emph{social welfare}
$w(s)$ of a state $s$ we use the natural utilitarian approach and
define $w(s) = \sum_{v \in V} w_v(s)$. A \emph{social optimum} $s^*$
is a state with $w(s^*) \ge w(s)$ for every possible state $s$ of the
game. Note that we restrict attention to states without randomization
and consider only pure Nash equilibria. In particular, the term ``Nash
equilibrium'' will only refer to the pure variant throughout the
paper.

In games such as ours, it makes sense to consider multilateral
deviations, as well as unilateral ones. Nash equilibria have
shortcomings in this context, for instance for a pair of adjacent
nodes who would -- although being unilaterally unable to increase
their utility -- benefit from cooperating and increasing the effort
jointly. The prediction of Nash equilibrium that such a state is
stable is quite unreasonable. In fact, it is easy to show that when
considering pure Nash equilibria, the function $\Phi(s) = w(s)/2$ is
an exact potential function for our games. This means that $s^*$ is an
optimal Nash equilibrium, the price of stability for Nash equilibria
is always 1, and iterative better response dynamics converge to an
equilibrium.  Additionally, for many natural reward functions $f_e$,
the price of anarchy for Nash equilibria remains
unbounded.\footnote{Consider, for instance, a path of length 3 with
  $(e_1,e_2,e_3)$ and $f_{e_1}(x,y)=f_{e_3}(x,y)=\min(x,y)$ and
  $f_{e_2}(x,y)=M\cdot\min(x,y)$, for some large number $M$. }

Following the reasoning in, for example,~\cite{Jackson96, Jackson08},
we instead consider pairwise equilibrium, and focus on the more
interesting case of \emph{bilateral deviations}. An improving
bilateral deviation in a state $s$ is a pair of strategies
$(s_u',s_v')$ such that $w_v(s_u',s_v',s_{-u,v}) > w_v(s)$ and
$w_u(s_u',s_v',s_{-u,v}) > w_u(s)$. A state $s$ is a \emph{pairwise
  equilibrium} if it is a Nash equilibrium and additionally there are
no improving bilateral deviations. Notice that we are actually using a
stronger notion of pairwise stability than described
in~\cite{Jackson08}, since any pair of players can change their
strategies in an arbitrary manner, instead of changing their
contributions on just a single edge. In particular, in a state $s$ a
coalition $C \subseteq V$ has a \emph{coalitional deviation} $s'_C$ if
the reward of every player in $C$ is strictly greater when all players
in $C$ switch from strategies $s_C$ to $s'_C$. $s$ is a \emph{strong
  equilibrium} if no coalition $C \subseteq V$ has a coalitional
deviation. Our notion of pairwise equilibrium is exactly the notion of
{\em 2-strong equilibrium}~\cite{Andelman09}, the restriction of
strong equilibrium to deviations of coalitions of size at most 2.

We evaluate the performance of stable states using prices of anarchy
and stability, respectively. The \emph{price of anarchy (stability)}
for pairwise equilibria in a game is the worst-case ratio of
$w(s^*)/w(s)$ for the worst (best) pairwise equilibrium $s$ in this
game. For a class of games (e.g., with certain convex reward
functions) that have pairwise equilibria, the price of anarchy
(stability) for pairwise equilibria is simply the worst price of
anarchy (stability) for pairwise equilibria of any game in the
class. If we consider classes of games, in which existence is not
guaranteed, the prices are defined as the worst prices of any game in
the class that has pairwise equilibria. Note that unless stated
otherwise, the terms price of anarchy and stability refer to pairwise
equilibria throughout the paper.

\subsection{Results and Contribution}
\label{sec:results}

We already observed above that in every game there always exist pure
Nash equilibria. In addition, iterative better response dynamics
converge to a pure Nash equilibrium, and the price of stability for
Nash equilibria is 1. The price of anarchy for Nash equilibria,
however, can be arbitrarily large, even for very simple reward
functions.

If we allow bilateral deviations, the conditions become much more
interesting. Consider the effort $s_v(e)$ expended by player $v$ on an
edge $e=(u,v)$. The fact that $f_e$ is monotonic nondecreasing tells us
that $w_v$ increases in $s_v(e)$. Depending on the application being
considered, however, the utility could possess the property of
``diminishing returns'', or on the contrary, could increase at a faster
rate as $v$ puts more effort on $e$. In other words, for a fixed effort
amount $s_u(e)$, $f_e$ as a function of $s_v(e)$ could be a concave or a
convex function, and we will distinguish the treatment of the framework
based on these properties.

\begin{table}[htb]
\label{tab:1}
\begin{center}
\renewcommand{\arraystretch}{1.5}
  \begin{tabular}{|c||c|c|}
    \hline

    & Existence & Price of Anarchy \\
    \hline\hline

    General convex&Yes (*)&2 \\

    \hline

    General concave& Not always &2\\

    \hline

    $c_e\cdot (x+y)$& Decision in \classP &1\\

    \hline
    Minimum effort convex&Yes (**)&2 (**)\\

    \hline
    Minimum effort concave&Yes&2\\

    \hline
    Maximum effort&Yes&2\\

    \hline
    Approximate Equilibrium & \multicolumn{2}{c|}{OPT is a 2-apx.\ Equilibrium}\\

    \hline

  \end{tabular}
  \caption{Summary of some of our results for various types of reward
    functions. For the cases where equilibrium always exists, we also
    give algorithms to compute it, as well as convergence results. All
    of our PoA upper bounds are tight. (*) If $\forall e$,
    $f_e(x,0)=0$, \classNP-hard otherwise. (**) If budgets are
    uniform, \classNP-hard otherwise.}
\end{center}
\end{table}

In Section~\ref{sec:convex} we consider the case of convex reward
functions. For a large class $\mathcal{C}$ of convex functions defined
below (Definition~\ref{def:classC}) we can show a tight bound for the
price of anarchy of 2 (Theorem~\ref{thm.classC}). However, for games
with functions from $\mathcal{C}$ pairwise equilibria might not
exist. In fact, we show that it is \classNP-hard to decide their
existence, even when the edges have simple reward functions of either
the form $f_e(x,y) = c_e \cdot (xy)$ or $f_e(x,y) = c_e \cdot (x+y)$
for constants $c_e > 0$ (Theorem~\ref{thm:generalHardness}). If,
however, {\em all} functions are of the form $f_e(x,y) = c_e\cdot
(xy)$, then existence and efficient computation are guaranteed. We
show this existence result for a substantially larger class of
functions that may not even be convex, although it includes the class
of all convex functions $f_e$ with $f_e(x,0)=0$
(Theorem~\ref{thm.convex}). Our procedure to construct a pairwise
equilibrium in this case actually results in a strong equilibrium,
i.e., the derived states are resilient to deviations of every possible
subset of players. As the prices of anarchy and stability for pairwise
equilibria are exactly 2, they extend to strong equilibria simply by
restriction.

As an interesting special case, we prove that if all functions are
$f_e(x,y) = c_e\cdot (x+y)$, it is possible to determine efficiently
if pairwise equilibria exist and to compute them in polynomial time in
the cases they exist (Theorem~\ref{thm:convexExists}).

In Section~\ref{sec:concave} we consider pairwise equilibria for
concave reward functions. In this case, pairwise equilibria may also
not exist. Nevertheless, in the cases when they exist, we can show
tight bounds of 2 on prices of anarchy and stability
(Theorem~\ref{thm:PoAconcave}).

Sections~\ref{sec:min} and~\ref{sec:max} treat different special cases
of particular interest. In Section~\ref{sec:min} we study the
important case of minimum effort games with reward functions $f_e(x,y)
= h_e(\min(x,y))$. If functions $h_e$ are convex, pairwise equilibria
do not necessarily exist, and it is \classNP-hard to decide the
existence for a given game (Theorem~\ref{thm:minHardness}). Perhaps
surprisingly, if budgets are uniform, i.e., if $B_v = B_u$ for all
$u,v \in V$, then pairwise equilibria exist for all convex functions
$h_e$ (Theorem~\ref{thm:convexMinExists}), and the prices of anarchy
and stability for pairwise equilibria are exactly 2
(Theorem~\ref{thm:PoAconvexMinUniform}). If functions $h_e$ are
concave, we can always guarantee existence
(Theorem~\ref{thm:concaveMinExists}). Our bounds for concave functions
in Section~\ref{sec:concave} imply tight bounds on the prices of
anarchy and stability of 2. Most results in this section extend to
strong equilibria. In fact, the arguments in all the existence proofs
can be adapted to show existence of strong equilibria, and tight
bounds on prices of anarchy and stability follow simply by
restriction.

In Section~\ref{sec:max} we briefly consider maximum effort games with
reward functions $f_e(x,y) = h_e(\max(x,y))$. For these games
bilateral deviations essentially reduce to unilateral ones. Hence,
pairwise equilibria exist, they can be found by iterative better
response using unilateral deviations, and the price of stability is 1
(Theorem~\ref{thm:maxExist}). In addition, we can show that the price
of anarchy is exactly 2, and this is tight
(Theorem~\ref{thm:maximum}).

Sections~\ref{sec:approx} to~\ref{sec:convergence} treat additional
aspects of pairwise equilibria. In Section~\ref{sec:approx} we
consider approximate equilibria and show that a social optimum $s^*$
is always a 2-approximate equilibrium (Theorem~\ref{thm.approxEq}). In
Section~\ref{sec:convergence} we consider sequential and concurrent
best response dynamics. We show that for general convex functions and
minimum effort games with concave functions the dynamics converge to
pairwise equilibria (Theorems~\ref{thm:convexConverge}
and~\ref{thm:concaveConverge}). For the former we can even provide a
polynomial upper bound on the convergence times.

Note that allmost all of our results on the price of anarchy for
pairwise equilibria result in a (tight) bound of 2. This bound of 2 is
essentially due to the dyadic nature of relationships, i.e., the fact
that edges are incident to at most two players. The case when edges
are projects among arbitrary subsets of actors is termed \emph{general
  contribution game} and treated in Section~\ref{sec:general}. Here we
consider setwise equilibria, which allow deviations by subsets of
players that are linked via a joint project. For some classes of such
games we show similar results for setwise equilibria as for pairwise
equilibria in network contributions games. In particular, we extend
the results on existence and price of anarchy for general convex
functions and minimum effort games with convex functions. The price of
anarchy for setwise equilibria becomes essentially $k$, where $k$ is
the cardinality of the largest project. However, many of the aspects
of this general case remain open, and we conclude the paper in
Section~\ref{sec:conclude} with this and other interesting avenues for
further research.

\subsection{Related Work}
\label{sec:relatedWork}

The model most related to ours is the co-author model~\cite{Jackson96,
  Jackson08}. The motivation of this model is very similar to ours,
although there are many important differences. For example, in the
usual co-author model, the nodes cannot choose how to split their
effort between their relationships, only which relationship to
participate in. Moreover, we consider general reward functions, and as
described above, our notion of pairwise stability is stronger than
in~\cite{Jackson96, Jackson08}.

Our games are potential games with respect to unilateral deviations
and can thus be embedded in the framework of congestion games. The
social quality of Nash equilibrium in non-splittable atomic congestion
games, where the quality is measured by social welfare instead of
social cost, has been studied in~\cite{Marden08}. Our games allow
players to split their effort arbitrarily between incident edges
(i.e., they are atomic {\em splittable} congestion
games~\cite{Orda93}), and we focus on coalitional equilibrium notions
like pairwise stability, not Nash equilibrium. In addition, the reward
functions (e.g., in minimum effort games) are much more general and
quite different from delay functions usually treated in the congestion
game literature~\cite{Bhaskar09,Cominetti09}.

In~\cite{BramoulleJET07}, Bramoull{\'e} and Kranton consider an
extremely general model of network games designed to model public
goods. Nevertheless, our game is not a special case of this model,
since in~\cite{BramoulleJET07} the strategy of a node is simply a
level of effort it contributes, not how much effort it contributes
{\em to each relationship.} There are many extensions to this model,
e.g., Corbo et al.~\cite{Corbo07} consider similar models in the
context of contributions in peer-to-peer networks. Their work closely
connects to the seminal paper on contribution games by Ballester et
al.~\cite{Ballester06}, which has prompted numerous similar follow-up
studies.

The literature on games played in networks is too diverse to survey
here -- we will address only the most relevant lines of research. In
the last few years, there have been several fascinating papers on
network bargaining games
(e.g.,~\cite{KleinbergSTOC08,ChakrabortyEC09}), and in general on
games played in networks where every edge represents a two-player game
(e.g.,~\cite{Davis09,DaskalakisMinMax09,HoeferDISC09}). All these
games either require that every node plays the same strategy on all
neighboring edges, or leaves the node free to play any strategy on any
edge. While every edge in our game can be considered to be a (very
simple) two-player game, the strategies/contributions that a node puts
on every edge are neither the same nor arbitrarily different:
specifically they are constrained by a budget on the total effort that
a node can contribute to all neighboring edges in total. To the best
of our knowledge, there have been no contributions (other than the
ones mentioned below) to the study of games of this type.

Our game bears some resemblance to network formation games where
players attempt to maximize different forms of network
centrality~\cite{Fabrikant03,Albers06,Brandes08,Kleinberg08,Laoutaris08,Bei09},
although our utility functions and equilibrium structure are very
different. Minimum effort coordination games as proposed by van Huyck
et al.~\cite{Huyck90} represent a special case of our general
model. They are a vital research topic in experimental economics, see
the papers mentioned above and~\cite{Devetag07} for a recent
survey. We study a generalized and networked variant in
Section~\ref{sec:min}. Slightly different adjustments to networks have
recently appeared in~\cite{AlosFerrerMini10,Bloch09}. Our work
complements this body of work with provable guarantees on the
efficiency of equilibria and the convergence times of dynamics.

Some of the special cases we consider are similar to stable
matching~\cite{Gusfield89}, and in fact correlated variants of stable
matching can be considered an ``integral'' version of our game. Our
results generalize existence and convergence results for correlated
stable matching (as, e.g., in~\cite{Ackermann11}), and our price of
anarchy results greatly generalize the results
of~\cite{AnshelevichD09}.

It is worth mentioning the connection of our reward functions with the
``Combinatorial Agency'' framework (see,
e.g.,~\cite{Babaioff06,BabaioffSAGT09}). In this framework, many
people work together on one project, and the success of this project
depends in a complex (usually probabilistic) manner on whether the
people involved choose a high level of effort. It is an interesting
open problem to extend our results to the case in which every project
of a game is an instance of the combinatorial agency problem.

A related but much more coordinated framework is studied in charity
auctions, which can be used to obtain contributions of rational agents for
charitable projects. This idea has been first explored by Conitzer and
Sandholm~\cite{Conitzer04}, and mechanisms for a social network setting
are presented by Ghosh and Mahdian~\cite{Ghosh08}.

\section{Polynomials and Convex Reward Functions}
\label{sec:convex}

In this section we start by considering a class of reward functions that
guarantee a small price of anarchy. 
We first introduce the notions of a \emph{coordinate-convex} and \emph{coordinate-concave} function.
\begin{definition}
  \label{def.coordinate}
  A function $f : \R^n\rightarrow \R$ is
  \begin{itemize}
  \item {\em coordinate-convex} if for all of its arguments $x_i$, we
    have that $\frac{\partial ^2 f}{\partial x_i^2}\geq 0$. A function
    is {\em strictly coordinate-convex} if all these are strict
    inequalities.
  \item {\em coordinate-concave} if for all of its arguments $x_i$, we
    have that $\frac{\partial ^2 f}{\partial x_i^2}\leq 0$. A function
    is {\em strictly coordinate-concave} if all these are strict
    inequalities.
  \end{itemize}
\end{definition}
Note that \emph{every convex function is coordinate-convex}, and
similarly every concave function is coordinate-concave. However,
coordinate-convexity/concavity is necessary but not sufficient for
convexity/concavity. For instance, the function $\log(1+xy)$ is
coordinate-concave, but not concave -- indeed, it is convex if $x = y
\in [0,1]$.

\begin{definition}
  \label{def:classC}
  Class $\mathcal{C}$ consists of all symmetric nondecreasing
  functions $f : \R_{\geq 0}\times \R_{\geq 0}\rightarrow \R_{\geq 0}$
  that are coordinate-convex. Define class $\mathcal{C}'$ as the
  subclass of functions from $\mathcal{C}$ that are strictly
  coordinate-convex. Define class $\mathcal{C}_0$ as the subclass of
  functions from $\mathcal{C}$ that satisfy $f(x,0) = 0$ for all $x
  \ge 0$.
\end{definition}

The class $\mathcal{C}$ is of particular interest to us, because we
can show the following result. The proof will appear below in
Section~\ref{sec:convexPoA}. \\

{\noindent \bf Theorem.} {\it For the class of network contribution
  games with reward functions $f_e \in \mathcal{C}$ for all $e \in E$
  that have a pairwise equilibrium, the prices of anarchy and
  stability for pairwise equilibria are exactly 2.} \\

Before we attack the proof, however, let us give some more intuition about
functions that belong to $\mathcal{C}$ and the properties of pairwise
equilibria in the corresponding games. Consider a polynomial $p(x,y)$ in
two variables with non-negative coefficients that is symmetric (i.e.,
$p(x,y) = p(y,x)$) and non-negative for $x,y \ge 0$.
For every such polynomial $p$ we consider all possible extensions to a
function $f(x,y) = h(p(x,y))$ with $h : \R_{\ge 0} \to \R_{\ge 0}$
being nondecreasing and convex. We call the union of all these
extensions the class $\mathcal{P}$. Clearly, every $p(x,y) \in
\mathcal{P}$ since $h(x) = x$ is convex. In particular, $\mathcal{P}$
contains a large variety of functions such as $xy$, $(x+y)^2$,
$e^{x+y}$, $x^3+y^3+2xy$, etc. Observe that $\mathcal{P} \subset
\mathcal{C}$, and thus the price of anarchy result for $\mathcal{C}$
will hold for every game with arbitrary functions from $\mathcal{P}$.
\begin{claim}
  It holds that $\mathcal{P} \subset \mathcal{C}$.
\end{claim}
\begin{proof}
  Let $f(x,y) = h(p(x,y))$ be an arbitrary function in $\mathcal{P}$
  as described above. $f$ is clearly monotone nondecreasing. 
  $\partial_{xx} p$, $\partial_{yy} p$, and $\partial_{xy} p$ are
  non-negative, since $p$ has positive coefficients.
  \[ \partial_{xx} f(x,y) = \partial_{pp} h(p(x,y))\cdot(\partial_x
  p(x,y))^2+\partial_p h(p(x,y))\cdot \partial_{xx} p(x,y) \geq 0
  \]
  since $h$ is convex. The same holds for the other second partial
  derivatives.
\end{proof}

\subsection{Existence and Computational Complexity of Pairwise
  Equilibria}

While we will show that the price of anarchy is 2 in
Section~\ref{sec:convexPoA}, this result says nothing about the
existence and complexity of computing pairwise equilibria. In fact,
even for simple games with reward functions $f_e(x,y) = c_e\cdot(x+y)$
and small constants $c_e$, pairwise equilibria can be absent.

\begin{example}
  \label{exm:noEq} \rm
  In our example there is a triangle graph with nodes $u_1$, $u_2$,
  and $u_3$, edges $e_1 = (u_1, u_2)$, $e_2 = (u_2, u_3)$, and $e_3 =
  (u_3, u_1)$, and uniform budgets. Edge $e_i$ has reward function
  $f_i$ with $f_1(x,y) = f_2(x,y) = 3(x+y)$, and $f_3(x,y) = 2(x+y)$.
  A pairwise equilibrium must not allow profitable unilateral
  deviations. Thus, $s_1(e_1) = s_3(e_2) = 1$, because this is
  obviously a dominant strategy w.r.t. unilateral deviations. Player 2
  can assign his budget arbitrarily. This yields $w_1(s) =
  3+3s_2(e_1)$ and $w_3(s) = 3+3s_2(e_2)$. Changing to a state $s'$
  where $u_1$ and $u_3$ bilaterally deviate by moving all their budget
  to $e_3$ yields $w_1(s') = 3s_2(e_1) + 4 > w_1(s)$ and $w_3(s') = 3s_2(e_2) + 4 > w_3(s)$. Hence, no
  pairwise equilibrium exists.
\end{example}

Although there are games without pairwise equilibria, there is a large
class of functions for which we can show existence and an efficient
algorithm for computation.

\begin{theorem}
  \label{thm.convex}
  A pairwise equilibrium always exists and can be computed efficiently
  when $f_e \in \mathcal{C}_0$ for all $e \in E$.
\end{theorem}

\begin{proof}
  We sort the edges in $E$ in decreasing order by maximum possible
  reward $c_{u,v} = f_{u,v}(B_u,B_v)$, and let $M$ be the result of a
  ``greedy matching'' algorithm for this order. Specifically, we add
  edges to $M$ in this decreasing order as long as adding the edge
  still results in a matching. This algorithm can be made to run in
  $O(m\log m)$ time. We now show that every state $s$ with
  $s_v(e)=B_v$ iff $e \in M$ is a pairwise equilibrium. The nodes that
  are not matched in $M$ can distribute their effort
  arbitrarily. Their payoff remains 0 since $f_e(x,0)=0$.

  We show the result by contradiction. First, suppose that a node $v$
  is willing to deviate unilaterally. Without loss of generality, we
  can assume this deviation removes effort from an edge of $M$, and
  adds all this effort to a single edge $e=(v,u)\not\in M$. We can
  assume this because when forming its best response, $v$ is
  maximizing the sum of convex functions under a budget constraint
  (since all reward functions are coordinate-convex). This means that
  whenever $v$ has an improving unilateral deviation where it adds
  effort to several edges, it also has an improving unilateral
  deviation where it adds all this effort to a single edge.

  For any edge $e=(v,u)\not\in M$ such that $u$ is matched in $M$,
  there is no reason for $v$ to add effort to $e$, since $s_u(e)=0$
  and $f_e(x,0)=0$.  If $u$ is not matched in $M$, then by moving $x$
  effort from edge $e'=(v,u')\in M$ to $e$, $v$ will obtain utility at
  most $f_e(x,B_u)+f_{e'}(B_v-x,B_{u'})$ instead of
  $f_e(0,B_u)+f_{e'}(B_v,B_{u'})$. This being an improving deviation
  implies that
  \begin{equation}\label{eqn:convexExistence.1}
    f_e(x,B_u)-f_e(0,B_u)>f_{e'}(B_v,B_{u'})-f_{e'}(B_v-x,B_{u'}).
  \end{equation}
  Define $e_v(y)=f_e(y,B_u)$ and $e'_v(y)=f_{e'}(y,B_{u'})$. Since
  $e'$ is chosen before $e$ by the greedy algorithm, it must be that
  $c_{e'}\geq c_e$, and so $e_v(B_v)\leq e'_v(B_v)$. Since $e_v(0) =
  e'_v(0) = 0$, there must be some interval of size $x$ on which
  $e'_v$ increases at least as much as $e_v$. But since both $e_v$ and
  $e'_v$ are convex, the interval $[0,x]$ must be the interval of
  smallest increase, and the interval $[B_v-x,B_v]$ is the interval of
  largest increase. This implies that
  \[ e_v(x)-e_v(0) \leq e'_v(B_v)-e'_v(B_v-x) \enspace, \]
  a contradiction with Inequality
  \ref{eqn:convexExistence.1}. Therefore, we only need to address
  bilateral deviations.

  Suppose that a node $v$ is willing to deviate by switching some $x$
  amount of its effort from edge $e' = (v,u') \in M$ to edge $e =
  (u,v)\not\in M$ as part of a bilateral deviation with $u$. We can
  assume w.l.o.g. that $e'$ was the first edge of $E_v\cup E_u$
  that was added to $M$, and so $c_{e'} \geq c_e$. For $v$ to be
  willing to deviate, it must be that Inequality
  \ref{eqn:convexExistence.1} is satisfied. The rest of the argument
  proceeds as before.
\end{proof}

Theorem~\ref{thm.convex} establishes existence and efficient
computation of equilibria for many functions from class
$\mathcal{C}$. In particular, it shows existence for all convex
functions $f_e$ that are 0-valued when one of its arguments is 0, as
well as for many non-convex ones, such as the weighted product
function $f_e(x,y)=c_e \cdot (xy)$. In fact, when considering
deviations of arbitrary coalitions of players, then it is easy to
verify that the player of the coalition incident to the edge with
maximum possible reward (of all edges incident to the players in the
coalition) does not make a strict improvement in the deviation. Thus,
as a corollary we get existence of strong equilibria.

\begin{cor}
  \label{cor.convex.strong}
  A strong equilibrium always exists and can be computed efficiently
  when $f_e \in \mathcal{C}_0$ for all $e \in E$.
\end{cor}

In general, we can show that deciding existence for pairwise
equilibria for a given game is \classNP-hard, even for very simple
reward functions from $\mathcal{C}$ such as $f_e(x,y) = c_e\cdot
(x+y)$ and $f(x,y) = c_e \cdot (xy)$ with constants $c_e > 0$.

\begin{theorem}
  \label{thm:generalHardness}
  It is \classNP-hard to decide if a network contribution game admits
  a pairwise equilibrium even if all functions are either $f_e(x,y) =
  c_e\cdot(x+y)$ or $f_e(x,y) = c_e\cdot(xy)$.
\end{theorem}

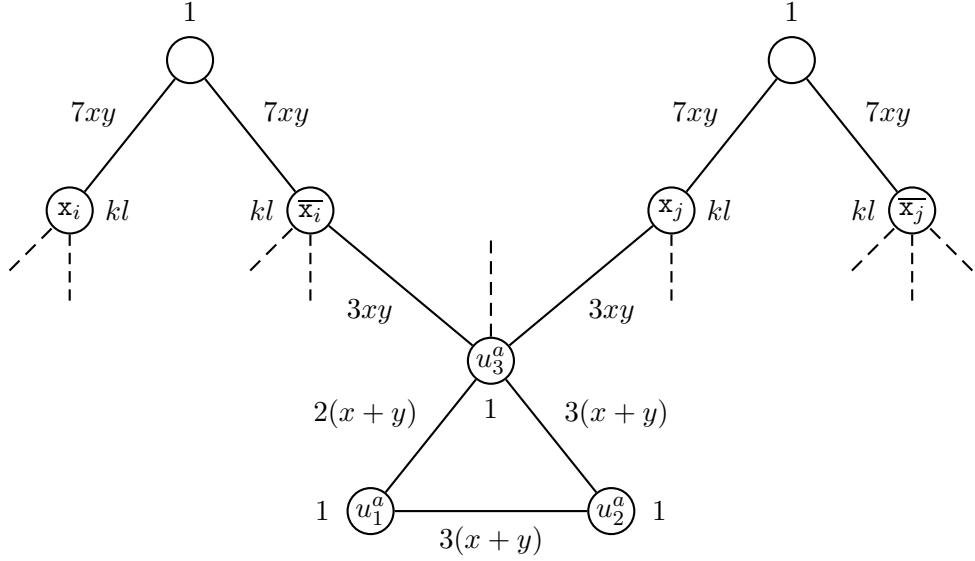
\begin{figure}
  \begin{center}
    \psset{unit=0.08}
    \begin{pspicture}(0,0)(200,90)

      \rput(30,60){\tt x$_i$}
      \cnode(30,60){4}{A}
      \rput(38,60){$kl$}

      \cnode(50,85){4}{B}
      \rput(50,93){$1$}

      \rput(70,60){$\overline{\text{\tt x$_i$}}$}
      \cnode(70,60){4}{C}
      \rput(62,60){$kl$}

      \rput(100,35){$u_3^a$}
      \cnode(100,35){4}{D}
      \rput(100,27){$1$}

      \rput(130,60){\tt x$_j$}
      \cnode(130,60){4}{E}
      \rput(138,60){$kl$}

      \cnode(150,85){4}{F}
      \rput(150,93){$1$}

      \rput(170,60){$\overline{\text{\tt x$_j$}}$}
      \cnode(170,60){4}{G}
      \rput(162,60){$kl$}

      \rput(80,10){$u_1^a$}
      \cnode(80,10){4}{H}
      \rput(72,10){$1$}

      \rput(120,10){$u_2^a$}
      \cnode(120,10){4}{I}
      \rput(128,10){$1$}

      \psline[linestyle=dashed](100,39)(100,55)
      \psline[linestyle=dashed](70,56)(70,45)
      \psline[linestyle=dashed](67,57)(60,50)

      \psline[linestyle=dashed](30,56)(30,45)
      \psline[linestyle=dashed](27,57)(20,50)
      \psline[linestyle=dashed](130,56)(130,45)
      \psline[linestyle=dashed](170,56)(170,45)
      \psline[linestyle=dashed](167,57)(160,50)
      \psline[linestyle=dashed](173,57)(180,50)

      \ncline[linestyle=solid]{A}{B}
      \naput{$7xy$}
      \ncline{B}{C}
      \naput{$7xy$}
      \ncline{C}{D}
      \nbput{$3xy$}
      \ncline{D}{E}
      \nbput{$3xy$}
      \ncline{E}{F}
      \naput{$7xy$}
      \ncline{F}{G}
      \naput{$7xy$}

      \ncline{D}{I}
      \naput{$3(x+y)$}
      \ncline{D}{H}
      \nbput{$2(x+y)$}
      \ncline{H}{I}
      \nbput{$3(x+y)$}
    \end{pspicture}
  \end{center}
  \caption{\label{fig:NPC} Construction for the \classNP-hardness
    proof. Labels inside vertices indicate role of the players, labels
    at vertices are budgets, and edge labels are reward
    functions. Vertices in the top layer are decision players with
    budget 1 corresponding to variables. For each decision player
    there are two adjacent assignment players (second layer) with
    budget $kl$ that indicate setting the decision variables. For each
    clause there is a triangle gadget (third and bottom layer) that by
    itself has no pairwise equilibrium. Connections between assignment
    players and triangle gadgets reflect the occurrences of variables
    in the clauses.}
\end{figure}

\begin{proof}
  We reduce from 3SAT as follows. We consider a 3SAT formula with $k$
  variables and $l$ clauses. For each clause we insert the game of
  Example~\ref{exm:noEq}. For each variable we introduce three players
  as follows. One is a \emph{decision player} that has budget 1. He is
  connected to two \emph{assignment players}, one true player and one
  false player. Both the true and the false player have a budget of $k
  \cdot l$. The edge between decision and assignment players has
  $f_e(x,y) = 7xy$. Finally, each assignment player is connected via
  an edge with $f_e(x,y) = 3xy$ to the player $u_3$ of every clause
  triangle, for which the corresponding clause has an occurrence of
  the corresponding variable in the corresponding form
  (non-negated/negated).

  Suppose the 3SAT instance has a satisfying assignment. We construct
  a pairwise equilibrium as follows. If the variable in the assignment
  is set true (false), we make the decision player contribute all his
  budget to the edge $e$ to the false (true) assignment player. This
  assignment player will contribute his full budget to $e$, because
  $7x$ has steeper slope than $3x$, which is the maximum slope
  attainable on the edges to the triangle gadgets. It is clear that
  none of these players has an incentive to deviate (alone or with a
  neighbor). The remaining set of assignment players $A$ can now
  contribute their complete budget towards the triangle gadgets. As
  the assignment is satisfying, every triangle player $u_3$ of the
  triangle gadgets has at least one neighboring assignment player in
  $A$. We now create a maximum bipartite matching between players in
  $A$ and the $u_3$ players of the triangles. We then extend this and
  connect the remaining (if any) triangle players arbitrarily to
  assignment players from $A$. This creates a one-to-many matching of
  triangle players to players in $A$, with every triangle player being
  matched to exactly one player in $A$, and some players in $A$ possibly unmatched. We set each triangle player to contribute all of his budget towards his edge in the matching. Each assignment player splits his effort evenly between the incident edges in the matching; if the assignment player is unmatched his strategy can be arbitrary. In this matching, each matched
  assignment player can get up to $l$ matching edges. As the triangle
  players contribute all their budget to their matching edge, then each edge in the matching yields a reward of $3x$, with $x$ being the contribution of the assignment player. By splitting his budget evenly, the assignment player
  contributes at least $k$ to each matching edge. Also he receives
  reward exactly $3kl$, which is the maximum achievable for a player
  in $A$ (given that the decision player does not contribute to the
  incident edge). Thus, every matched assignment player in $A$ is
  stable and will not join a bilateral deviation. Consider a triangle
  player $u_3$. As the assignment player he is matched to contributes
  at least $k$ (we assume w.l.o.g.\ $k \ge 3$), the reward function on
  the matching edge grows at least as quickly as $9y$, i.e., with a
  larger slope than the maximum slope achievable on the triangle
  edges. In addition, the reward for $u_3$ by contributing all budget to the matching edge is at least 9. Note that the maximum payoff that he can obtain by contributing only to triangle edges is 8, and therefore he has no
  incentive to join other triangle players in a deviation. Note that
  $u_3$ could potentially achieve higher revenue by deviating with a different assignment player in $A$. However, as noted above no
  matched assignment player has an incentive to deviate jointly with
  $u_3$. Hence, $u_3$ can only join an unmatched assignment
  player. This is only a profitable deviation if $u_3$ currently
  shares his current assignment player with at least one other
  triangle player. However, the possibility that $u_3$ could deviate
  to such an unmatched assignment player contradicts the fact that we
  created a maximum matching between assignment and triangle
  players. Thus, $u_3$ will also stick to his strategy
  choice, and has no incentive to participate in bilateral or unilateral deviations. We can stabilize the remaining pairs of triangle
  players by assigning an effort of 1 towards their joint edge. Finally, the unmatched assignment players $v$ in $A$ are stable since their reward is always 0: no player adjacent to $v$ puts any effort on edges incident to $v$, and no player adjacent to $v$ is willing to participate in a bilateral deviation due to the arguments above.

  Now suppose there is a pairwise equilibrium. Note first that the
  decision player will always contribute his full budget, and there is
  always a positive contribution of at least one assignment player
  towards the decision player edge -- otherwise there is a joint
  deviation that yields higher reward for both players. In particular,
  the decision player contributes only to edges, where the maximum
  contribution of the assignment players is located. As the decision
  player contributes his full budget, there is at least one incident
  edge that grows at least as quickly as $3.5x$ in the contribution $x$
  of the assignment player. Hence, at least one assignment player will
  be motivated to remove all contributions from the edges to the
  triangle players, as these edges grow at most by $3x$ in the
  contribution $x$ of the assignment player. He will instead invest
  all of his budget towards the decision player. This implies that
  every pairwise equilibrium must result in a decision for the
  variable, i.e., if the (false) true assignment players contributes
  all of his budget towards the decision player, the variable is set
  (true) false. If both players do this, the variable can be chosen
  freely. As there is a stable state, the contributions of the
  remaining assignment players must stabilize all triangle gadgets. In
  particular, this means that for each clause triangle there must be
  at least one neighboring assignment player that does not contribute
  all of his budget towards his decision player. This implies a
  satisfying assignment for the 3SAT instance.
\end{proof}

Finally, let us focus on an interesting special case. The hardness in
the previous theorem comes from the interplay of reward functions $xy$
that tend to a clustering of effort and $x+y$ that create cycles. We
observed above that if all functions are $c_e \cdot(xy)$, then
equilibria exist and can be computed efficiently. Here we show that
for the case that $f_e(x,y) = c_e\cdot(x+y)$ for all $e \in E$, we can
decide efficiently if a pairwise equilibrium exists. Furthermore, if
an equilibrium exists, we can compute it in polynomial time.

\begin{theorem}\label{thm:convexExists}
  There is an efficient algorithm to decide the existence of a
  pairwise equilibrium, and to compute one if one exists, when all
  reward functions are of the form $f_e(x,y) = c_e \cdot (x+y)$ for
  arbitrary constants $c_e > 0$. Moreover, the price of anarchy is 1
  in this case.
\end{theorem}

\begin{proof}
  Let $S^*$ be the set of socially optimum solutions. These are
  exactly the solutions where every player $v$ puts effort only on
  edges with maximum $c_e$. $S^*$ are exactly the solutions that are
  stable against unilateral deviations, which immediately tells us
  that if a pairwise equilibrium exists, then the price of anarchy is
  1.

  Not all solutions in $S^*$ are stable against bilateral deviations,
  however. Denote $c_v = \max_{e \in E_v} c_e$. Let $E_v^*$ be the set
  of edges incident to $v$ with value $c_v$. In any unilaterally
  stable solution, a node $v$ must put all of its effort on edges in
  $E_v^*$. We first show how to determine if a pairwise stable
  solution exists if the only edges in the graph are $\cup_v E_v^*$.

  Consider an edge $e=(u,v)$ such that $e\in E_u^*$ but $e\not\in
  E_v^*$ (call such an edge "Type 1"). Then in any pairwise stable
  solution, the node $u$ must contribute {\em all} of its effort to
  edge $e$. Otherwise $u$ and $v$ could deviate by $v$ adding some
  amount $\e>0$ to $e$, and $u$ adding $>\e(c_v-c_u)/c_u$ to $e$. This
  is possible for small enough $\e$, and would improve the reward for
  both $u$ (by $\e c_u$) and $v$ (by
  $>\e(c_v-c_u)-\e(c_v-c_u)=0$). Therefore, for every edge $e=(u,v)$
  of this type, we can fix the contributions of node $u$, since they
  will be the same in any stable solution. If the same node $u$ has
  two or more such incident edges $e$, then by the above argument we
  immediately know that there does not exist any pairwise equilibrium.

  Now consider edges $e=(u,v)$ which are in $E_u^*\cap E_v^*$, which
  implies that $c_u=c_v$. For any such edge, either $c_u(e)=B_u$ or
  $c_v(e)=B_v$ in any pairwise equilibrium. If this were not the case,
  then both $u$ and $v$ could add some $\e$ amount of effort to $e$
  and benefit from this deviation by $2\e c_u$ amount. Consider a
  connected component consisting of such edges. We can use simple flow
  or matching arguments to find if there exists an assignment of nodes
  to edges such that every edge has at least one adjacent node
  assigned to it. We then set $c_u(e)=B_u$ if node $u$ is assigned to
  edge $e$. We also make sure not to assign a node that already used
  its budget on a Type 1 edge to {\em any} edge in this phase. As
  argued above, if such an assignment does not exist, then there is no
  pairwise equilibrium. Conversely, any such assignment yields a
  pairwise equilibrium, since for every edge, at least one of the
  endpoints of this edge is using all of its effort on this edge. Thus
  we are able to determine exactly when pairwise equilibria exist on
  the set of edges $\cup_v E_v^*$. Call the set of such solutions $S$.

  All that is left to check is if one of these solutions is stable
  with respect to bilateral deviations on edges $e\not\in \cup_v
  E_v^*$. If one solution in $S$ is a pairwise equilibrium in the
  entire graph, then all of them are, since when moving effort onto an
  edge $e$, a node does not care which edge it removes the effort
  from: all the edges with positive effort have the same slope. To
  verify that a pairwise equilibrium exists, we simply consider every
  edge $e=(u,v)\not\in \cup_v E_v^*$ with reward function $c_e\cdot
  (x+y)$, and check if $(c_u-c_e)(c_v-c_e)\geq c_e^2$. We claim that a
  pairwise equilibrium exists iff this is true for all edges.

  Consider a bilateral deviation onto edge $e$ where $u$ contributes
  $\e_1$ effort and $v$ contributes $\e_2$ effort. This would be an
  improving deviation exactly when $c_e\e_2>(c_u-c_e)\e_1$ and
  $c_e\e_1>(c_v-c_e)\e_2$. Fix $\e_1>0$ to be some arbitrarily small
  value; then there exists $\e_2$ satisfying the above conditions
  exactly when $(c_u-c_e)/c_e<c_e/(c_v-c_e)$, which is true exactly
  when $(c_u-c_e)(c_v-c_e)< c_e^2$, as desired.
\end{proof}


\subsection{Price of Anarchy}
\label{sec:convexPoA}

This section is devoted to proving the following theorem.

\begin{theorem}
  \label{thm.classC}
  For the class of network contribution games with reward functions
  $f_e \in \mathcal{C}$ for all $e \in E$ that have a pairwise
  equilibrium, the prices of anarchy and stability for pairwise
  equilibria are exactly 2.
\end{theorem}

We will refer to an edge $e=(u,v)$ as being {\em slack} if
$B_u>s_u(e)>0$ and $B_v>s_v(e)>0$, {\em half-slack} if $B_u>s_u(e)>0$
but $s_v(e)\in\{B_v,0\}$, and {\em tight} if $s_u(e)\in\{B_u,0\}$ and
$s_v(e)\in\{B_v,0\}$. We will call a solution {\em tight} if it has
only tight edges.


\begin{claim}
  \label{claim.optIntegral}
  If all reward functions belong to class $\mathcal{C}$, then there
  always exists a tight optimum solution. If all reward functions
  belong to class $\mathcal{C}'$, then all optimum solutions are
  tight.
\end{claim}

\begin{proof}
  Let $s$ be a solution with maximum social welfare, and let node $v$
  be a node that uses non-zero effort on two adjacent edges: $e=(u,v)$
  and $e'=(w,v)$. For simplicity, we will denote $f_e$ by $f$ and
  $f_{e'}$ by $g$. Furthermore, we denote $s_u(e)$ by $\alpha_u$,
  $s_v(e)$ by $\alpha_v$, $s_v(e')$ by $\beta_v$, $s_w(e')$ by
  $\beta_w$. The fact that $v$ switching an $\varepsilon$ amount of
  effort from $e$ to $e'$ or from $e'$ to $e$ does not increase the
  social welfare means that:

  \begin{equation}
    f(\alpha_v+\varepsilon,\alpha_u)-f(\alpha_v,\alpha_u) \leq g(\beta_v,\beta_w)- g(\beta_v-\varepsilon,\beta_w)
  \end{equation}

  and that

  \begin{equation}
   g(\beta_v+\varepsilon,\beta_w) - g(\beta_v,\beta_w) \leq f(\alpha_v,\alpha_u) - f(\alpha_v-\varepsilon,\alpha_u).
  \end{equation}

  We know from coordinate-convexity that $g(\beta_v,\beta_w) -
  g(\beta_v-\varepsilon,\beta_w) \leq g(\beta_v+\varepsilon,\beta_w) -
  g(\beta_v,\beta_w)$. Therefore, we have that

  \begin{equation}
    f(\alpha_v+\varepsilon,\alpha_u) - f(\alpha_v,\alpha_u) \leq f(\alpha_v,\alpha_u) - f(\alpha_v-\varepsilon,\alpha_u).
  \end{equation}

  This is not possible if $f$ is in $\mathcal{C}'$, giving us a
  contradiction, and completing the proof for $f\in\mathcal{C}'$. For
  $f\in\mathcal{C}$, this tells us that $v$ moving its effort from one
  of these edges to the other will not change the social welfare, and
  so we can create an optimum solution with one less half-slack edge
  by setting $\varepsilon=\min(\alpha_v,\beta_v)$. We can continue
  this process to end up with a tight optimum solution, as desired.
\end{proof}


\begin{extraproof}{Theorem~\ref{thm.classC}} Let $s$ be a pairwise
  stable solution, and $s^*$ an optimum solution. By
  Claim~\ref{claim.optIntegral} we can assume that $s^*$ is tight.


  Define $w_e(s)$ to be the reward of edge $e$ in $s$, and $w_e(s^*)$
  to be the reward of $e$ in $s^*$. Recall that for a node $v$, the
  utility of $v$ is $w_v(s)=\sum_{e\in E_v}
  w_e(s)$.
  %
  %
  Let $e=(u,v)$ be an arbitrary tight edge in $s^*$. If $s^*_u(e)=B_u$
  and $s^*_v(e)=B_v$, then consider the bilateral deviation from $s$
  where both $u$ and $v$ put all their effort on edge $e$. Since $s$
  is pairwise stable, there must be some node (wlog node $u$) such
  that $w_u(s)\geq w_e(s^*)=f_e(B_u,B_v)$. Make this node $u$ a
  witness for edge $e$. If instead $s^*_u(e)=B_u$ and $s^*_v(e)=0$,
  then consider the unilateral deviation from $s$ where $u$ puts all
  its effort on edge $e$. Since $s$ is stable against unilateral
  deviation, then $w_u(s)\geq w_e(s^*)=f_e(B_u,0)$. Make this node $u$
  a witness for edge $e$.

  Notice that every node can be a witness for at most one edge, since
  for a node $u$ to be a witness to edge $e$, it must be that
  $s^*_u(e)=B_u$.  Therefore, we know that $\sum_v w_v(s)\geq \sum_e
  w_e(s^*)$. Since the total social welfare in $s^*$ is exactly
  $2\sum_e w_e(s^*)$, we know that the price of anarchy for pairwise
  equilibria is at most 2.

  Finally, let us establish tightness of this bound. Consider a path
  of four nodes with uniform budgets and edges $e_1 = (u,v)$,$e_2 =
  (v,w)$ and $e_3 = (w,z)$. The reward functions are $f_{e_1}(x,y) =
  f_{e_3}(x,y) = xy$ and $f_{e_2}(x,y) = (1+\varepsilon)xy$. $v$ and
  $w$ achieve their maximum reward by contributing their full budget
  to $e_2$, hence they will apply this strategy in every pairwise
  equilibrium. This leaves no reward for $u$ and $z$ and gives a total
  welfare of $2+2\varepsilon$. If the players contribute only to $e_1$
  and $e_3$, the total welfare is $4$. Hence, the price of stability
  for pairwise equilibria is at least 2, which matches the upper bound
  on the price of anarchy.
\end{extraproof}

For completeness, we also present a result similar to Claim \ref{claim.optIntegral} for pairwise equilibrium solutions.

\begin{claim}
  \label{claim.stableIntegral}
  If all reward functions belong to class $\mathcal{C}$, with every reward function $f_e$ having the property that $\frac{\partial^2 f_e}{\partial x\partial y}\geq 0$, then for every pairwise equilibrium, there exists a pairwise equilibrium of the same welfare without slack edges.
\end{claim}

\begin{proof}
  Let $s$ be a pairwise stable solution, and suppose it contains a
  slack edge $e=(u,v)$.


  This means that nodes $u$ and $v$ have other adjacent edges where
  they are contributing non-zero effort. Let those edges be
  $e_1=(u,w_1)$ and $e_2=(v,w_2)$. For simplicity, we will denote
  $f_e$, by $f$, $f_{e_1}$ by $f_1$, and $f_{e_2}$ by
  $f_2$. Furthermore, we denote $s_v(e)$ by $\alpha_v$, $s_u(e)$ by
  $\alpha_u$, $s_u(e_1)$ by $\beta_u$, $s_v(e_2)$ by $\gamma_v$ and
  for $w_1$ and $w_2$ accordingly.

  For any value $\e\leq\min\{\alpha_u,\beta_u\}$, it must be that $u$
  cannot unilaterally deviate by moving $\e$ effort from $e$ to $e_1$,
  or from $e_1$ to $e$. Therefore, we know that
  $f(\alpha_u+\e,\alpha_v) - f(\alpha_u,\alpha_v)\leq
  f_1(\beta_u,\beta_{w_1})- f_1(\beta_u-\e,\beta_{w_1})$, and that
  $f_1(\beta_u+\e,\beta_{w_1}) - f_1(\beta_u,\beta_{w_1})\leq
  f(\alpha_u,\alpha_v) - f(\alpha_u-\e,\alpha_v)$. Since $f$ and $f_1$
  are coordinate-convex, however, we know that $f(\alpha_u,\alpha_v) -
  f(\alpha_u-\e,\alpha_v)\leq f(\alpha_u+\e,\alpha_v) -
  f(\alpha_u,\alpha_v)$ (and similarly for $f_1$), which implies that
  the above inequalities hold with equality.  Specifically, it implies
  that for both $f$ and $f_1$, increasing $u$'s effort by $\e$ causes
  the same difference in utility as decreasing it by $\e$. This simply
  quantifies the fact that for a node to put effort on more than one
  edge in a stable solution, it should be indifferent between those
  two edges.

  For any $\e\leq\min\{\alpha_v,\alpha_u,\beta_u,\gamma_v\}$, consider
  the pairwise deviation where $u$ and $v$ both move $\e$ amount of
  effort to $e$ from $e_1$ and $e_2$. Suppose w.l.o.g.\ that node $u$
  is not willing to deviate in this manner because it does not
  increase its utility. This means that

  \begin{equation}\label{stableIntegral:eqn1}
    f(\alpha_u+\e, \alpha_v+\e) - f(\alpha_u,\alpha_v) \leq
    f_1(\beta_u,\beta_{w_1}) - f_1(\beta_u-\e,\beta_{w_1}).
  \end{equation}

  By the above argument about $u$'s unilateral deviations, we know
  that $f_1(\beta_u,\beta_{w_1}) - f_1(\beta_u-\e,\beta_{w_1}) =
  f(\alpha_u+\e,\alpha_v) - f(\alpha_u,\alpha_v)$, and so we have that
  $f(\alpha_u+\e, \alpha_v+\e) = f(\alpha_u+\e, \alpha_v)$. Since $\frac{\partial^2 f}{\partial x\partial y}\geq 0$, this also implies that
  $f(\alpha_u, \alpha_v+\e) = f(\alpha_u, \alpha_v)$.

  We now create solution $s'$ by having node $v$ move $\e$ effort from
  edge $e_2$ to $e$. We will prove below that $s'$ is also a pairwise
  stable solution of the same welfare as $s$. However, it has strictly
  more effort on edge $e$. We can continue applying the same arguments
  for edge $e$ until $e$ is no longer a slack edge. This process only
  decreases the number of slack edges, since we only remove effort
  from edges that are slack or half-slack, and in the latter case we
  remove the effort from the "slack" direction, so the edge remains
  half-slack afterwards.

  All that is left to prove is that $s'$ is pairwise stable of the
  same welfare as $s$. To see that the welfare is the same, notice
  that we can use for node $v$ the same arguments for unilateral
  deviations that we applied to $u$.  Therefore, we know that
  \begin{eqnarray*}
    f(\alpha_u,\alpha_v+\e) - f(\alpha_u,\alpha_v)=
    f_2(\gamma_{w_2},\gamma_v) - f_2(\gamma_{w_2},\gamma_v-\e) = 0.
  \end{eqnarray*}
  Therefore, the reward of all edges in $s'$, and thus the utility of all nodes, is the same as in $s$, so the social welfare of both is the same.

  We must now prove that $s'$ is pairwise stable. Any possibly
  improving deviation would have to include one of the edges $e=(u,v)$
  or $e_2=(v,w_2)$, since for all other edges the effort levels are
  the same in $s$ and $s'$. First consider deviations (unilateral or
  bilateral) including node $v$. Any such deviation would also be a
  valid deviation in $s$, since only the strategy of node $v$ has
  changed. Therefore, all such deviations cannot be improving
  deviations. Next consider any unilateral deviation by $u$ where $u$
  adds some $\delta$ effort to edge $e$. For this to be a strictly
  improving deviation, it must be that $f(\alpha_u+\delta,\alpha_v+\e)
  - f(\alpha_u,\alpha_v+\e)>0$.  Consider instead a bilateral
  deviation from $s$ where $u$ plays the new strategy (i.e., adds
  $\delta$ to $e$), while $v$ deviates by moving $\e$ from $e_2$ to
  $e$. In this deviation, $u$ strictly benefits since it ends in the
  same configuration as above, and $v$ strictly benefits since it
  loses $f_2(\gamma_{w_2},\gamma_v) - f_2(\gamma_{w_2},\gamma_v-\e) =
  0$ utility and gains $f(\alpha_u+\delta,\alpha_v+\e) -
  f(\alpha_u,\alpha_v)>0$ utility.  Therefore, this contradicts $s$
  being pairwise stable.

  Next consider a deviation from $s'$ where $u$ removes some amount
  $\delta$ from $e$. For this deviation to be profitable in $s'$, but
  not profitable in $s$, it must be that $f(\alpha_u,\alpha_v+\e) -
  f(\alpha_u-\delta,\alpha_v+\e) < f(\alpha_u,\alpha_v) -
  f(\alpha_u-\delta,\alpha_v).$ This contradicts the fact that $\frac{\partial^2 f}{\partial x\partial y}\geq 0$.

  Finally, consider deviations by node $w_2$. If $w_2$ deviates and
  removes some amount $\delta$ from $e_2$, then this is profitable
  only if $f_2(\gamma_{w_2},\gamma_v-\e) -
  f_2(\gamma_{w_2}-\delta,\gamma_v-\e) < f_2(\gamma_{w_2},\gamma_v) -
  f_2(\gamma_{w_2}-\delta,\gamma_v).$ Recall, however, that
  $f_2(\gamma_{w_2},\gamma_v) = f_2(\gamma_{w_2},\gamma_v-\e),$ so the
  above implies that $f_2(\gamma_{w_2}-\delta,\gamma_v-\e) >
  f_2(\gamma_{w_2}-\delta,\gamma_v)$, which is impossible since $f_2$
  is nondecreasing in both its arguments. If $w_2$ adds effort to
  $e_2$ in its deviation, then it cannot possibly be more profitable
  than the same deviation in $s$, since $v$ is using less effort on
  $e_2$ in $s'$ than in $s$, and the utility of edge $e_1$ is the same
  in both. This finishes the proof.
\end{proof}

\begin{cor}
  If all reward functions belong to class $\mathcal{C}'$, then all
  pairwise equilibria have only tight edges.
\end{cor}

\begin{proof}
  In the proof of Claim \ref{claim.stableIntegral}, we saw that if a
  node is putting non-zero effort on two edges, then it must be that
  for some edge $e$ and values $x,y$, we have that
  $f_e(x+\e,y)-f_e(x,y)=f_e(x,y)-f_e(x-\e,y)$. This is not possible
  for $f_e\in\mathcal{C'}$, since $f_e$ is strictly convex in each of
  its arguments.
\end{proof}

\begin{cor}
  If all reward functions belong to class $\mathcal{C}$, then all
  strict pairwise equilibria (where every player has a unique
  unilateral best response) have only tight edges.
\end{cor}

\begin{proof}
  In the proof of Claim \ref{claim.stableIntegral}, we saw that if a
  node is putting non-zero effort on two edges, then its utility does
  not change by moving some amount of effort from one of these edges
  to the other. This is not possible in a strict pairwise equilibrium,
  since then this node would have a deviation that does not change its
  utility.
\end{proof}


\section{Concave Reward Functions}
\label{sec:concave}

In this section we consider the case when reward functions $f_e(x,y)$
are concave. It is simple to observe that a pairwise equilibrium may
not exist. Consider a triangle graph with three players, uniform
budgets, and $f_e(x,y) = \sqrt{xy}$ for all edges. Every player has an
incentive to invest his full budget due to monotonic increasing
functions. Due to concavity each player will even out the
contributions according to the derivatives. Thus, the only candidate
for a pairwise equilibrium is when all players put 0.5 on each
incident edge. It is, however, easy to see that this state is no
pairwise equilibrium. Although we might have no pairwise equilibrium,
we obtain the following general result for games with concave rewards
that have a pairwise equilibrium.

\begin{theorem}
  \label{thm:PoAconcave}
  For the class of network contribution games with concave reward
  functions for all $e \in E$ that have a pairwise equilibrium, the
  price of anarchy for pairwise equilibria is at most 2.
\end{theorem}

\begin{proof}
  Consider a social optimum $s^*$, and the effort $s_v^*(e)$ used by
  node $v$ on edge $e$ in this solution. Let $s$ be a pairwise
  equilibrium, with $s_v(e)$ the effort used by $v$ on edge $e$ in
  $s$. For an edge $e=(u,v)$, let $w_e(s)=f_e(s_u(e),s_v(e))$ be its
  reward in $s$, and $w_e(s^*)=f_e(s^*_u(e),s^*_v(e))$ be its reward in
  $s^*$. We will now attempt to charge $w(s^*)$ to $w(s)$.

  For any node $v$, define $O^v$ to be the set of edges incident to $v$
  where $v$ contributes strictly more in $s^*$ than in $s$, i.e., where
  $s_v^*(e)> s_v(e)$. Similarly, define $S^v$ to be the set of edges
  $e$ where $s_v^*(e)\leq s_v(e)$.

  Let $O$ be the set of edges $e$ with strictly higher reward in $s^*$ than in
  $s$ ($w_e(s^*)> w_e(s)$), and $S$ be the rest of the edges in the
  graph, with reward in $s$ at least as high as in $s^*$. Furthermore, define $O_1=\{e=(u,v)|e\in O^u\cap S^v\cap O\}$
  and $O_2=\{e=(u,v)|e\in O^u\cap O^v\cap O\}$. In other words, $O_2$
  is the set of edges with higher reward in $s^*$ where both players
  contribute more in $s^*$ than in $s$, and $O_1$ is the set of edges
  where only one player does so. Similarly, define $S_1=\{e=(u,v)|e\in
  O^u\cap S^v\cap S\}$ and $S_2=\{e=(u,v)|e\in S^u\cap S^v\cap
  S\}$. Since reward functions are monotone, every edge must appear in
  exactly one of $O_1$, $O_2$, $S_1$, or $S_2$.

  In the following proof, we will first show that any edge $e=(u,v)$
  in $O$ can be {\em assigned} to one of its endpoints (say $u$) such
  that $u$ would never gain in deviating from $s$ by removing effort
  from edges in $S^u$ and making its contribution to $e$ equal to
  $s^*_u(e)$, even if $v$ did the same. This means that the utility
  node $u$ looses from setting its contribution to $s_u^*(e)$ instead
  of $s_u(e)$ on all edges of $S^u$ is at least as much as the
  difference in utility in $s^*$ versus in $s$ on all the edges of $O$
  assigned to $u$. We then sum up these inequalities, which lets us
  bound the reward on edges where $s^*$ is better than $s$ by the
  reward on edges where $s$ is better than $s^*$. We now proceed with
  the proof as described. 

  Let $e=(u,v)$ be an arbitrary edge of $O$, so $w_e(s^*)>w_e(s)$.
  Since $f_e$ is nondecreasing, this implies that $s_v(e) < s_v^*(e)$
  or $s_u(e) < s_u^*(e)$, i.e., at least one of $u$ or $v$ has strictly lower
  effort on $e$ in $s$ than in $s^*$. Consider the deviation
  from $s$ to another state $s'$ where $u$ and $v$ increase their
  contributions to $e$ to the same level as in $s^*$, i.e., a state
  $s'$ yields $s_u'(e)=\max\{s_u(e),s^*_u(e)\}$ and
  $s_v'(e)=\max\{s_v(e),s^*_v(e)\}$. This may be either a bilateral or a unilateral deviation, depending on whether one of $s_u'(e)=s_u(e)$ or $s_v'(e)=s_v(e)$ holds. Note that there is actually an
  entire set of such states $s'$, as we did not specify from where
  players $u$ and $v$ potentially remove effort to be able to achieve
  the increase. Observe, however, that no other player changes his
  strategy, i.e., $s'_{-u,v} = s_{-u,v}$. Since $s$ is an equilibrium,
  it must be that for at least one of $u$ or $v$ the deviation to
  \emph{every possible} such state $s'$ is unprofitable. Without loss
  of generality, say that this player is $u$, so $w_u(s') \leq w_u(s)$
  for every state $s'$, and we say that we \emph{assign} edge $e$ to
  node $u$. Note that this implies that $\delta_u(e)=s^*_u(e)-s_u(e)>0$, i.e.,
  $e\in O^u$, since otherwise all edges incident to $u$ would have the same reward in every $s'$ as in $s$, except for the edge $e$ which would have reward $w_e(s^*)$ in $s'$, strictly greater than $w_e(s)$.

  Since $f_e(s_u'(e),s_v'(e)) \geq w_e(s^*)$, then there is an increase in
  $w_u$ due to $e$ of at least $w_e(s^*)-w_e(s)$. However, as every
  deviation to a state $s'$ is unprofitable for $u$, it must be that
  removing $\delta_u(e)$ effort in \emph{any arbitrary way} from other
  edges incident to $u$ and adding it to $e$ would not increase $u$'s
  utility $w_u$. Therefore, we know that, in particular, removing
  $\delta_u(e)$ effort from edges $S^u$ decreases the reward of those
  edges by at least $w_e(s^*)-w_e(s)$. Denote by
  $\chi_u(\delta_u(e))$ this amount, i.e, $\chi_u(\delta)$ is the
  minimum amount that $w_u$ would decrease if in state $s$ player $u$
  removed any $\delta$ amount of effort from edges in $S^u$.

  We have now proven that for any $e=(u,v)\in O$, we can assign it to one of its endpoints (say $u$), such that $\chi_u(\delta_u(e))\geq
  w_e(s^*)-w_e(s)$. We can now sum these inequalities for every
  edge $e\in O$. Consider the sum of just the inequalities
  corresponding to the edges assigned to a fixed node $v$ (call this
  set of edges $A(v)$). Then we have that
  \[ \sum_{e\in A(v)}[w_e(s^*)-w_e(s)] \leq
  \sum_{e\in A(v)}\chi_v(\delta_v(e))\enspace.\]
  How does $\chi_v(\delta_1)+\chi_v(\delta_2)$ compare to
  $\chi_v(\delta_1+\delta_2)$?  Since all the functions $f_e$ are
  concave, it is easy to see that removing $\delta_1+\delta_2$ effort from
  edges $S_v$ will decrease the reward of these edges by at least as
  much as the sum of $\chi_v(\delta_1)$ and
  $\chi_v(\delta_2)$. Therefore, we know that
  \[ \sum_{e\in A(v)}\chi_v(\delta_v(e)) \leq \chi_v\left(\sum_{e\in
      A(v)}\delta_v(e)\right)\enspace. \]
  Since $\delta_v(e)$ is the extra effort of $v$ on edge $e$ in $s^*$
  compared to $s$, the sum of $\delta_v(e)$ for the edges $O^v$ equals
  $\Delta = \sum_{e\in S^v}s_v(e)-s_v^*(e)$. Thus, $\chi_v(\Delta)$
  is at most the utility lost by $v$ if, starting at state $s$,
  $v$ would set its contribution to $s_v^*(e)$ instead of $s_v(e)$ on all edges of $S^v$.
  For an edge $e\in S_2$, this is at most $w_e(s)-w_e(s^*)$, since even
  after lowering $v$'s contribution to $s_v^*(e)$, the reward of this edge
  is at least $w_e(s^*)$. For an edge $e\in S_1$ or $e\in O_1$, this is still at
  most $w_e(s)$. Noticing that an edge of $S^v$ cannot be in $O_2$, we now
  have that
  \[\chi_v(\Delta) \leq \sum_{e\in S^v\cap S_2}[w_e(s)-w_e(s^*)] +
  \sum_{e\in S^v\cap(S_1\cup O_1)}w_e(s)\enspace.\]
  Putting this all together, we obtain that
  \[ \sum_{e\in A(v)}[w_e(s^*)-w_e(s)] \leq
  \sum_{e\in S^v\cap S_2}[w_e(s)-w_e(s^*)] +
  \sum_{e\in S^v\cap(S_1\cup O_1)}w_e(s)\enspace.\]
  Summing up these inequalities for all nodes $v$, we obtain a way to
  bound the reward on edges where $s^*$ is better than $s$ by the
  reward on edges where $s$ is better than $s^*$. Since the same edge
  $e=(u,v)$ could be in both $S^u$ and $S^v$, it may be used in the
  above sum twice. Notice, however, that any edge in $S_1$ or $O_1$ will
  only appear in this sum once, since it will belong to $S^v$ of exactly
  one node. Thus, we obtain that
  \[ \sum_{e\in O} [w_e(s^*)-w_e(s)] \leq2\sum_{e\in S_2}[w_e(s)-w_e(s^*)] +
  \sum_{e\in S_1\cup O_1}w_e(s)\enspace.\]
  Adding in the edges of $S_1$, and recalling that all edges are in exactly one of $O_1$, $O_2$, $S_1$, or $S_2$,
   gives us the desired bound:
  \[ w(s^*)\leq 2\sum_{e\in S_2}w_e(s) +
  \sum_{e\in S_1\cup O_1}w_e(s) +\sum_{e\in O}w_e(s)+\sum_{e\in S_1}w_e(s)\leq 2w(s)\enspace.\]
\end{proof}
Looking carefully at the proof of the previous theorem yields the
following result (c.f.~Definition~\ref{def.coordinate}).
\begin{cor}
  \label{cor:PoAconcave}
  For the class of network contribution games with coordinate-concave
  reward functions for all $e \in E$ that have a pairwise equilibrium,
  the price of anarchy for pairwise equilibria is at most 2.
\end{cor}

\section{Minimum Effort Games}
\label{sec:min}

In this section we consider the interesting case (studied for example
in~\cite{Huyck90,Anderson01,Fatas06,Riechmann08,Dufwenberg05,Chaudhuri08})
when all reward functions are of the form $f_e(x,y)=h_e(\min(x,y))$. In
other words, the reward of an edge depends only on the minimum effort of
its two endpoints. In our treatment we again distinguish between the case
of increasing marginal returns (convex functions $h_e$) and diminishing
marginal returns (concave functions $h_e$). Note that in this case
bilateral deviations are in many ways essential to make the game
meaningful, as there is almost always an infinite number of Nash
equilibria.\footnote{In particular, due to
  monotonic increasing functions any state in which for each edge the
  contributions of incident players are the same is a Nash
  equilibrium.} In addition, we can assume w.l.o.g.\ that in every
pairwise equilibrium $s$ there is a unique value $s_e$ for each $e =
(u,v) \in E$ such that $s_v(e) = s_u(e) = s_e$. The same can be
assumed for optima $s^*$.

%
%
%

We begin by showing a simple yet elegant proof based on linear programming
duality, that shows a price of anarchy of 2 when all functions $h_e(x) =
c_e \cdot x$ are linear with slope $c_e > 0$. We include this proof to
highlight that duality is also used in
Theorem~\ref{thm:PoAconvexMinUniform} for convex functions and uniform
budgets.

\begin{theorem}
  \label{thm:MinLinear}
  The prices of anarchy and stability for pairwise equilibria in games
  with all functions of the form $f_e(x,y) = c_e \cdot \min(x,y)$ are
  exactly 2.
\end{theorem}

\begin{proof}
  We use linear programming duality to obtain the result. Consider an
  arbitrary pairwise equilibrium $s$ and an optimum $s^*$. Note that
  the problem of finding $s^*$ can be formulated as the following
  linear program, with variables $x_e$ representing the minimum
  contribution to edge $e$:
  \newcommand{\D}{\displaystyle}
  \begin{equation}
    \begin{array}{lll}
      \mbox{Max } & \D \sum_{e \in E} 2 c_e x_e \\
      \mbox{s.t. } & \D \sum_{e : e=(u,v)} x_e \le B_u & \mbox{ for all } u \in V \\
      & x_e \ge 0\enspace.
    \end{array}
  \end{equation}
  The LP-dual of this program is
  \begin{equation}
    \begin{array}{lll}
      \mbox{Min } & \D \sum_{u \in V} B_u y_u \\
      \mbox{s.t. } & \D y_u + y_v \ge 2c_e & \mbox{ for all } e \in E \\
      & y_u \ge 0\enspace.
    \end{array}
  \end{equation}
  Now consider the pairwise equilibrium $s$ and a candidate dual
  solution $y$ composed of
  \[
  y_u = \sum_{e : e=(u,v)} \frac{c_e s_e}{B_u}.
  \]
  If a player contributes all of his budget in $s$, this is the
  average payoff per unit of effort. Note that $\{s_e\}$ is a feasible
  primal, and $\sum_e 2c_e s_e = \sum_u y_u B_u$, but $\{y_u\}$ is not
  a feasible dual solution. Now suppose that for an edge $e$ both
  incident players $u$ and $v$ have $y_u, y_v < c_e$. Then both
  incident players can either move effort from an edge with
  below-average payoff to $e$, or invest some of their remaining
  budget on $e$. This increases both their payoffs and contradicts
  that $s$ is stable. Thus, for every edge $e$ there is a player $u$
  with $y_u \ge c_e$. Thus, by setting $y'_u = 2y_u$ we obtain a
  feasible dual solution with profit of twice the profit of $s$. The
  upper bound follows by standard duality arguments. It is
  straightforward to derive a tight lower bound on the price of
  stability using a path of length 3 and functions $h_e(x) = x$ and
  $h_e(x) = (1+\varepsilon)x$ in a similar fashion as presented in
  Theorem~\ref{thm.classC} previously.
\end{proof}

\subsection{Convex Functions in Minimum Effort Games}

In this section we consider reward functions $f_e(x,y) =
h_e(\min(x,y))$ with convex functions $h_e(x)$. This case bears some
similarities with our treatment of the class $\mathcal{C}$ in
Section~\ref{sec:convex}. In fact, we can show existence of pairwise
equilibria in games with uniform budgets. We call an equilibrium $s$
\emph{integral} if $s_e \in \{0,1\}$ for all $e \in E$.

\begin{theorem}
  \label{thm:convexMinExists}
  A pairwise equilibrium always exists in games with uniform budgets
  and $f_e(x,y) = h_e(\min(x,y))$ when all $h_e$ are convex. If all
  $h_e$ are strictly convex, all pairwise equilibria are integral.
\end{theorem}

\begin{proof}
  We first show how to construct a pairwise equilibrium. The proof is
  basically again an adaptation of the ``greedy matching'' argument
  that was used to show existence for general convex functions in
  Theorem~\ref{thm:convexExists}. In the beginning all players are asleep. We
  iteratively wake up the pair of sleeping players that achieves the
  highest revenue on a joint edge and assign them to contribute their
  total budget towards this edge. The algorithm stops when there is no
  pair of incident sleeping players.

  Suppose for contradiction that the resulting assignment is not a
  pairwise equilibrium. First consider a bilateral deviation, where
  a pair of players can profit from
  re-assigning some budget to an edge $e'$. By our algorithm at least
  one of the players incident to $e'$ is awake. Consider the incident
  player $u$ that was woken up earlier. If it is profitable for him to
  remove some portion $x$ of effort from an edge $e$ to $e'$, this
  implies
  \[
  h_e(1) < h_e(1-x) + h_{e'}(x)
  \]
  However, our choices imply $h_e(1) \ge h_{e'}(1)$. Convexity yields
  $h_{e'}(x) \le x h_{e'}(1) $ and $h_e(1-x) \le (1-x)h_e(1)$ and
  results in a contradiction
  \begin{eqnarray*}
  h_e(1) & < & h_e(1-x) + h_{e'}(x)\\
  &\le& (1-x)h_e(1) + xh_{e'}(1)\\
  &\le& (1-x)h_e(1) + xh_e(1) \\
  & = & h_e(1)\enspace.
  \end{eqnarray*}
  This implies that the algorithm computes a stable state with respect
  to bilateral deviations. As for unilateral deviations, no player
  would ever add any effort to an edge where the other endpoint is
  putting in zero effort. However, if a player $u$ unilaterally
  re-assigns some $x$ budget to an edge $e'=(u,v)$ from edge $e$ with
  $v$ still being asleep at the end of the algorithm, then this
  implies that $h_e(1)<h_e(1-x)+h_{e'}(x)$ and that $h_e(1)\geq
  h_{e'}(1)$. This gives a contradiction by the same argument as
  above.

  If all functions are strictly convex, then $h_e(x) < x h_e(1)$ for
  all $x \in (0,1)$. In this case we show that every stable state $s$
  is integral, i.e., we have $s_e \in \{0,1\}$. Suppose to the
  contrary that there is an equilibrium $s$ with $s_e \in (0,1)$ for
  $e=(u,v)$. Let $e$ be an edge with the largest value $h_e(1)$ such
  that $s_e\in (0,1)$. For player $u$, let $e_i$ with $i=1,...$ be
  other incident edges of $u$ such that $s_{e_i} \in (0,1)$. Then,
  because of strict convexity, we have
  \[
  h_e(s_e) + \sum_i h_{e_i}(s_{e_i}) < s_e h_e(1) + \sum_i s_{e_i}
  h_{e_i}(1) \le h_e(1) \enspace.
  \]
  This means $u$ has an incentive to move all of his effort to $e$ if
  $v$ does the same. By the same argument, $v$ also has an incentive
  to move all its effort to $e$. Thus, the bilateral deviation of $u$
  and $v$ moving their effort to $e$ is an improving deviation for
  both $u$ and $v$, so we have a contradiction to $s$ being stable.
\end{proof}


\begin{theorem}\label{thm:PoAconvexMinUniform}
  The prices of anarchy and stability for pairwise equilibria in
  network contribution games are exactly 2 when all reward functions
  $f_e(x,y) = h_e(\min(x,y))$ with convex $h_e$, and budgets are
  uniform.
\end{theorem}

\begin{proof}
  Consider a stable solution $s$ and an optimum solution $s^*$. For a
  vertex $u$ we consider the profit and denote this by $y_u = \sum_{e
    : e=(u,v)} f_e(s_e).$ For every edge $e=(u,v)$, consider the case
  when both players invest the full effort. Due to convexity $s_e^*
  \cdot f_e(1) \ge f_e(s_e^*)$. Suppose for both players $y_u, y_v <
  f_e(1)$. Then there is a profitable switch by allocating all effort
  to $e$. This implies that $\max\{y_u, y_v\} \ge f_e(1)$ and thus,
  \[ (y_u + y_v) \cdot s_e^* \ge s_e^* f_e(1) \ge f_e(s_e^*) \enspace. \]
  Thus, we can bound
  \[
  \sum_{e \in E} (y_u + y_v)\cdot s_e^* \ge \sum_{e \in E} f_e(s_e^*)
  = w(s^*)/2 \enspace. \]
  On the other hand
  \[
  \sum_{e \in E} (y_u + y_v) \cdot s_e^* = \sum_{u \in V} \sum_{e :
    e=(u,v)} y_u \cdot s_e^* \le \sum_{u \in V} y_u = 2 \sum_{e \in E}
  f_e(s_e) = w(s) \enspace.
  \]
  Hence, $w(s) \ge w(s^*)/2$ and the price of anarchy is 2.

  Note that this is tight for functions $f$ that are arbitrarily
  convex. The example is a path of length 3 similar to
  Theorem~\ref{thm.classC} and Theorem~\ref{thm:MinLinear}. We use
  $f_e(1) = 1+\varepsilon$ for the inner edge and $f_e(1) = 1$ for the
  outer edges, and the price of stability becomes arbitrarily close to
  2.
\end{proof}

For the case of arbitrary budgets and convex functions, however, we
can again find an example that does not allow a pairwise equilibrium.

\begin{example}
  \label{exm:minNoEq} \rm
  Our example game consists of a path of length 3. We denote the
  vertices along this path with $u$, $v$, $w$, $z$. All players have
  budget 2, except for player $z$ that has budget 1. The profit
  functions are $h_{u,v}(x) = 2x^2$, $h_{v,w}(x) = 5x$, and
  $h_{w,z}(x) = 6x$. Observe that this game allows no pairwise
  equilibrium: If $2 \ge s_{v,w} > 1$, then player $w$ has an
  incentive to increase the effort towards $z$. If $1 \ge s_{v,w} >
  0$, then player $v$ has an incentive to increase effort towards
  $u$. If $s_{v,w} = 0$, both $v$ and $w$ can jointly increase their
  profits by contributing $2$ on $(v,w)$.
\end{example}

Using this example we can construct games in which deciding existence of
pairwise equilibria is hard.

\begin{theorem}
  \label{thm:minHardness}
  It is \classNP-hard to decide if a network contribution game admits
  a pairwise equilibrium if budgets are arbitrary and all
  functions are $f_e(x,y) = h_e(\min(x,y))$ with convex $h_e$.
\end{theorem}

\begin{proof}
  We reduce from 3SAT and use a similar reduction to the one given in
  Theorem \ref{thm:generalHardness}.  An instance of 3SAT is given by
  $k$ variables and $l$ clauses. For each clause we construct a simple
  game of Example~\ref{exm:minNoEq} that has no stable state. For each
  variable we introduce three players as follows. One is a
  \emph{decision player} that has budget $k \cdot l$.  He is connected
  to two \emph{assignment players}, one true player and one false
  player. Both these players have also a budget of $k \cdot l$. The
  edge between decision and assignment players has $h_e(x) =
  10x^2$. Finally, each assignment player is connected via an edge
  with $h_e(x) = 7x$ to the node $z$ of every clause path, for which
  the corresponding clause has an occurrence of the corresponding
  variable in the corresponding form (non-negated/negated). Note that
  the connecting player $z$ is the only player with budget 1 in the
  clause path.

  Suppose the 3SAT instance has a satisfying assignment. We construct
  a stable state as follows. If the variable is set true (false), we
  make the decision player contribute all his budget to the edge $e$
  to the false (true) assignment player. Both assignment player and
  decision player are motivated to contribute their full budget to
  $e$, because $10(kl)^2$ is the maximum profit that they will ever be
  able to obtain. Clearly, none of these players has an incentive to
  deviate (alone or with a neighbor). The remaining set of assignment
  players $A$ can now contribute their complete budget towards the
  clause gadgets. As the assignment is satisfying, every node $z$ of
  the clause gadgets has at least one neighboring assignment player in
  $A$. We create a maximum bipartite matching of clause players $z$ to
  players in $A$ and match the remaining clause players $z$ (if any)
  to players from $A$ arbitrarily. Each clause player $z$ contributes
  all of his budget towards his edge in this one-to-many
  matching. Each assignment player splits his effort evenly between
  the incident edges in the matching. Note that the players $z$ from
  the clause gadgets now receive profit 7, which is the maximum
  achievable. Thus, they have no incentive to deviate. Hence, no
  player in $A$ has a profitable unilateral or a possible bilateral
  deviation. Finally, we obtain a stable state in the clause gadgets
  by assigning all players $v$ and $w$ to contribute 2 to $(v,w)$.

  Now suppose there is a stable state. Note first that the decision
  player and one incident assignment player can and will obtain their
  maximum profit by contributing their full budgets towards a joint
  edge -- otherwise there is a joint deviation that yields higher
  profit for both players. Hence, this assignment player will not
  contribute to edges to the clause gadget players. This implies a
  decision for the variable, i.e., if the (false) true assignment
  players contributes all of his budget towards the decision player,
  the variable is set (true) false. As there is a stable state, the
  contributions of the remaining assignment players must stabilize all
  clause gadgets. In particular, this means that for each clause
  triangle there must be at least one neighboring assignment player
  that does not contribute towards his decision player. This implies
  that the assignment decisions made by the decision players must be
  satisfying for the 3SAT instance.
\end{proof}

The construction of Example~\ref{exm:minNoEq} and the previous proof
can be extended to show hardness for games with uniform budgets
in which functions are either concave or convex.

\begin{cor}
  In games with uniform budgets and functions $f_e(x,y) = h_e(\min(x,y))$ with
  monotonic increasing $h_e$ it is \classNP-hard to determine if a
  pairwise equilibrium exists.
\end{cor}

\begin{proof}
  We use the same approach as in the previous proof, however, we
  assign each player a budget of $B_u = k \cdot l$. For each of the
  players $u$, $v$, $w$, and $z$ in a clause gadget we introduce
  players $u'$, $v'$, $w'$ and $z'$. $u'$ is only connected to $u$,
  $v'$ only to $v$, and similar for $w'$ and $z'$. The edges $(u,u')$,
  $(v,v')$ and $(w,w')$ have profit function $h_e(x) =
  10\cdot(kl-2)^{1.5} \cdot \sqrt{x}$ for $x \le kl-2$ and $h_e(x) =
  10(kl-2)^2$ otherwise. Similarly, we use $h_e(x) =
  10\cdot(kl-1)^{1.5} \cdot\sqrt{x}$ for $x \le kl-1$ and $h_e(x) =
  10(kl-1)^2$ otherwise for $(z,z')$. It is easy to observe that in
  every pairwise equilibrium players $u$ and $u'$ will contribute
  $kl-2$ towards their joint edge. This holds accordingly for every
  other pair of players $(v,v')$, $(w,w')$ and $(z,z')$. The remaining
  budgets of the players are the budgets used in
  Example~\ref{exm:minNoEq} above and lead to the same arguments in
  the above outlined reduction.
\end{proof}
Finally, we observe that the existence result in
Theorem~\ref{thm:convexMinExists} extends to strong equilibria. In
particular, whenever we consider a deviation from a coalition of
players, the reward of players incident to the highest reward edge do
not strictly improve by the deviation. In addition, the prices of
anarchy and stability are 2 because our lower bound examples continue
to hold for strong equilibria, while the upper bounds follow by
restriction.
\begin{cor}
  \label{cor:convexMinStrong}
  A strong equilibrium always exists in games with uniform budgets and
  $f_e(x,y) = h_e(\min(x,y))$ when all $h_e$ are convex. If all $h_e$
  are strictly convex, all strong equilibria are integral. The prices
  of anarchy and stability for strong equilibria in these games are
  exactly 2.
\end{cor}

\subsection{Concave Functions in Minimum Effort Games}

In this section we consider the case of diminishing returns, i.e., when
all $h_e$ are concave functions. Note that in this case the function
$f_e = h_e(\min(x,y))$ is coordinate-concave. Therefore, the results
from Section~\ref{sec:concave} show that the price of anarchy is at
most 2. However, for general coordinate-concave functions it is not
possible to establish the existence of pairwise equilibria, which we
do for concave $h_e$ below. In fact, if the functions $h_e$ are
strictly concave, we can show that the equilibrium is unique.
\begin{theorem}
  \label{thm:concaveMinExists}
  A pairwise equilibrium always exists in games with $f_e(x,y) =
  h_e(\min(x,y))$ when all $h_e$ are continuous, piecewise
  differentiable, and concave. It is possible to compute pairwise
  equilibria efficiently within any desired precision. Moreover, if
  all $h_e$ are strictly concave, then this equilibrium is unique.
\end{theorem}
\begin{proof}
  First, notice that we can assume without loss of generality that for
  every edge $e=(u,v)$, the function $h_e$ is constant for values
  greater than $\min(B_u,B_v)$. This is because it will never be able
  to reach those values in any solution.

  We create a pairwise equilibrium in an iterative manner. For any
  solution and set of nodes $S$, define $BR_v(S)$ as the set of best
  responses for node $v$ {\em if it can control the strategies of
    nodes $S$.} We begin by computing $BR_v(V)$ independently for each
  player $v$ ($V$ is the set of all nodes). In particular, this
  simulates that $v$ is the player that always creates the minimum of
  every edge, and we pick $s_v$ such that it maximizes $\sum_{e =
    (u,v)} h_e(s_v(e))$. This is a concave maximization problem (or
  equivalently a convex minimization problem), for which it is
  possible to find a solution by standard methods in time polynomial
  in the size of $G$, the encoding of the budgets $B_v$ and the number
  of bits of precision desired for representing the solution. For
  background on efficient algorithms for convex minimization see,
  e.g.,~\cite{Nesterov94}.

  Let $h_e^+(x)$ be the derivative of $h_e(x)$ in the positive
  direction, and $h_e^-(x)$ be the derivative of $h_e(x)$ in the
  negative direction. We have the property that for $s_v$ calculated
  as above, for every edge $e$ with $s_v(e) > 0$ it holds that
  $h_e^-(s_v(e))\ge h_{e'}^+(s_v(e'))$ for every edge $e'$ incident to
  $v$. Define $h_v'$ as the minimum value of $h_e^-(s_v(e))$ for all
  edges $e$ incident to $v$ with $s_v(e) > 0$.

  Our algorithm proceeds as follows. At the start all players are
  asleep, and in each iteration we pick one player to wake up. Let
  $S_i$ denote the set of sleeping players in iteration $i$, and $A_i
  = V - S_i$ the set of awake players; in the beginning $S_1=V$. We
  will call edges with both endpoints asleep {\em sleeping} edges, and
  all other edges {\em awake} edges.

  In each iteration $i$, we pick one player to wake up, and fix its
  contributions on all of its adjacent edges. In particular, we choose
  a node $v \in S_i$ with the currently highest derivative value
  $h'_v$ (see below for tie-breaking rule). We set $v$'s contribution
  to an edge $e=(u,v)$ to $s_v(e)$, where $s_v\in BR_v(S_i)$. Define
  $\overline{BR}_v(S_i)$ as the set of best responses in $BR_v(S_i)$
  for which $s_v(e) = s_u(e)$ for all awake edges $e = (u,v)$. For
  $s_v \in \overline{BR}_v(S_i)$ player $v$ exactly matches the
  contributions of the awake nodes $A_i$ on all awake edges between
  $v$ and $A_i$. By Lemma~\ref{lem:matchingContribution} below,
  $\overline{BR}_v(S_i)$ is non-empty, and our algorithm sets the
  contributions of $v$ to $s_v\in \overline{BR}_v(S_i)$. Moreover, we
  set the contribution of other sleeping players $u \in S_i$ to be
  $s_u(e) = s_v(e)$ on the sleeping edges, so we assume $u$ fully
  matches $v$'s contribution on edge $e$. By
  Lemma~\ref{lem:matchingContribution}, $u$ will not change its
  contributions on these edges when it is woken up. Thus, in the final
  solution output by the algorithm $v$ will receive exactly the reward
  of $\overline{BR}_v(S_i)$. Now that node $v$ is awake, we compute
  $\overline{BR}_u(S_i-\{v\})$ for all sleeping $u$, as well as new
  values $h_u'$ and iterate. Note that values $h_u'$ in later
  iterations are defined as the minimum derivative values on all the
  {\em sleeping} edges neighboring $u$, not on all edges. To
  summarize, each iteration $i$ of the algorithm proceeds as follows:

  \begin{itemize}
  \item For every $u\in S_i$, compute $s_u\in \overline{BR}_u(S_i)$.
  \item For every $u\in S_i$, set $h_u'$ to be the minimum value of $h_e^-(s_u(e))$ for all {\em sleeping} edges $e$ incident to $u$ with $s_u(e) > 0$.
  \item Choose a node $v$ with maximum $h'_v$ (using tie-breaking rule below), fix $v$'s strategy to be $s_v$, and set $S_{i+1}=S_i-\{v\}$.
  \end{itemize}

  To fully specify the algorithm, we need to define a tie-breaking
  rule for choosing a node to wake up when there are several nodes
  with equal values $h_v'$. Let $s_v \in \overline{BR}_v(S_i)$ that we
  compute. Our goal is that for every edge $e=(u,v)$ with $h_u' =
  h_v'$, we choose node $u$ such that $s_u(e)\leq s_v(e)$. We claim
  that we can always find a node $u$ such that this is true with
  respect to all its neighbors. Suppose a node $u$ has two edges
  $e=(u,v)$ and $e'=(u,w)$ with $h_u'=h_v'=h_w'$ and $s_u(e)>s_v(e)$
  but $s_u(e')<s_w(e')$. Lemma~\ref{lem:sameDerivative} below implies
  that the functions on $(u,v)$ and $(u,w)$ are linear in this
  range. Specifically, Lemma \ref{lem:sameDerivative} implies that
  $h_{e'}^+(s_u(e'))=h_e^-(s_u(e))$ because
  $h_e^-(s_u(e))=h_u'=h_w'=h_{e'}^-(s_w(e'))\leq h_{e'}^+(s_u(e'))$,
  with the inequality being true because $h_{e'}$ is concave. Hence,
  $u$ can move some amount of effort from $e$ to $e'$ and still form a
  best response. Continuing in this manner, we can find another best
  response in $\overline{BR}_u(S_i)$ for $u$ such that $u$ has
  contributions that are either more than both its neighbors, or less
  than both its neighbors. This implies that there exists $u \in S_i$
  with $s_u \in \overline{BR}_u(S_i)$ such that $s_u(e)\leq s_v(e)$
  for all neighbors $v$, and therefore our tie-breaking is possible.

  \begin{lemma}\label{lem:sameDerivative}
    Consider two nodes $u$ and $v$ and an edge $e=(u,v)$, and let
    $s_u\in \overline{BR}_u(S_i)$ and $s_v\in \overline{BR}_v(S_i)$ be
    the best responses computed in our algorithm. Suppose that
    $s_v(e)>s_u(e)$. Then it must be that either $h_u'>h_v'$, or
    $h_u'=h_e^-(s_v(e))=h_v'$.
  \end{lemma}

  \begin{proof}
    If edge $e$ is the edge which achieves the minimum value $h_u'$,
    then we are done, since then $h_u'=h^-_e(s_u(e))\geq
    h^-_e(s_v(e))\geq h_v'$.  Therefore, we can assume that another
    edge $e'=(u,w)$ with $s_u(e')>0$ achieves this value, so $h_u' =
    h^-_{e'}(s_u(e'))$.

    The fact that we cannot increase $u$'s reward by assigning more
    effort to edge $e$ means that $h^-_{e'}(s_u(e'))\geq
    h^+_e(s_u(e))$. Since $h_e$ is concave, we know that $h^+_e(s_u(e))\geq
    h^-_e(s_v(e))$, which is at least $h_v'$ by its definition. This
    proves that $h_u'\geq h_v'$. If this is a strict inequality, then
    we are done. The only possible way that $h_u'= h_v'$ is if $h_u'=
    h^-_{e'}(s_u(e'))= h^-_e(s_v(e)) = h_v'$, as desired.
  \end{proof}

  First we will prove that our algorithm forms a feasible solution,
  i.e., that the budget constraints are never violated. To do this, we
  must show that when the $i$'th node $v$ is woken up and sets its
  contribution $s_v(e)$ on a newly awake edge $e=(u,v)$, the other
  sleeping player $u$ must have enough available budget to match
  $s_v(e)$. In $s_u \in \overline{BR}_u(S_i)$ that our algorithm
  computes, let $\overline{B_u}$ be the available budget of node $u$,
  that is,
  \[
  \overline{B_u} = B_u - \sum_{e=(u,w), w \in A_i}
  s_w(e)\enspace,
  \]
  the budget minus requested contributions on awake edges. This is the
  maximum amount that node $u$ could assign to $e$.

  For contradiction, assume that $s_v(e) > \overline{B_u}$, so our
  assignment is infeasible. Then it must be that $h_e^-(s_u(e))\geq
  h_e^-(\overline{B_u})\geq h_e^-(s_v(e))$, since $h_e$ is concave. By
  definition of $h_v'$, we know that $h_e^-(s_v(e))\geq h_v'$, and so
  $h_e^-(s_u(e))\geq h_v'$. Now let $e'=(u,w)$ be the edge that
  achieves the value $h_u'$, i.e., $h_{e'}^-(s_u(e'))=h_u'$. If
  $h_{e'}^-(s_u(e'))<h_v'$, then
  $h_{e'}^-(s_u(e'))<h_e^-(\overline{B_u})\leq h_e^+(s_u(e))$, so $s_u$
  cannot be a best response, since $u$ could earn more reward by
  switching some amount of effort from $e'$ to $e$. Therefore, we know
  that $h_u'\geq h_v'$. If this is a strict inequality, then we have a
  contradiction, since $u$ would have been woken up before
  $v$. Therefore, it must be that $h_u'=h_v'$. But this contradicts
  our tie-breaking rule -- we would choose $u$ before $v$ because it
  puts less effort onto edge $e$ in our choice from
  $\overline{BR}_u(S_i)$ than $v$ does in
  $\overline{BR}_v(S_i)$. Therefore, our algorithm creates a
  feasible solution.

  \begin{lemma}\label{lem:matchingContribution}
    For every node $u$ and all $S_i$ until node $u$ is woken up, there
    is a best response in $BR_u(S_i)$ that exactly matches the
    contributions of the awake nodes $A_i$. In other words,
    $\overline{BR}_u(S_i)$ is non-empty.
  \end{lemma}

  \begin{proof}
    We prove this by induction on $i$; this is trivially true for
    $S_1$. Suppose this is true for $S_{i-1}$, and let $v$ be the
    node that is woken up in the $i$'th iteration, with an existing
    edge $e=(u,v)$, so that $S_i=S_{i-1}-\{v\}$. Let
    $s_u\in \overline{BR}_u(S_{i-1})$ be $u$'s best response which
    exists by the inductive hypothesis. First, we claim that
    $s_u(e)\geq s_v(e)$. To see this, notice that if $s_u(e)< s_v(e)$,
    then by Lemma~\ref{lem:sameDerivative}, we know that $h_u'\geq h_v'$.
    If this is a strict inequality, then we immediately get a contradiction, since
    we picked $v$ to wake up because it had the highest $h_v'$
    value. If $h_u'=h_v'$, this contradicts our tie-breaking rule,
    since $u$ would be woken up first for contributing less to edge
    $e$.

    Consider the computation of $BR_u(S_{i-1})$ from $u$'s point of view. $u$ is deciding how to allocate
    its budget $B_u$ among incident edges, in order to maximize its
    reward. By putting $x$ effort onto an edge $e'=(u,w)$ with $w\in S_{i-1}$,
    $u$ will obtain $h_{e'}(x)$ reward, since $u$ can control the strategy
    of $w$, and so will make it match the contribution of $u$ on edge
    $e'$. If instead $w\in A_{i-1}$, then by putting $x$ effort onto $e'$, $u$
    will only obtain $h_{e'}(\min(x,s_w(e')))$ reward, since the strategy
    of $w$ is already fixed, and $u$ cannot change it. Then, $BR_u(S_{i-1})$
    is simply the set of budget allocations of $u$ that maximizes the sum
    of the above reward functions. Now consider the computation of
    $BR_u(S_i)$ and compare it to $BR_u(S_{i-1})$. The only difference is
    that $u$ cannot control the node $v$ when computing $BR_u(S_i)$, i.e.,
    by putting $x$ effort onto edge $e$, node $u$ will only obtain
    $h_e(\min(x,s_v(e)))$ reward, instead of $h_e(x)$.


    If $s_u(e)=s_v(e)$, then $s_u\in BR_u(S_i)$ as well as in $BR_u(S_{i-1})$,
    since the computations of $BR_u(S_i)$ and $BR_u(S_{i-1})$ only
    differ in the reward function of edge $e$, and $u$ cannot gain any utility by
    putting more than $s_u(e)$ effort onto edge $e$ in $BR_u(S_i)$. $s_u$
    matches all the contributions of nodes in $A_i$ (including $v$), and
    so $\overline{BR}_u(S_i)$ is non-empty.

    Suppose instead that $s_u(e)>s_v(e)$. Now, let $s'_u$ be a strategy of
    $u$ created from $s_u$ as follows. Remove effort from edge $e$ by setting
    $s'_u(e)=s_v(e)$, and add $s_u(e)-s_v(e)$ effort to the other
    edges of $u$ in the optimum way to maximize $u$'s utility in $BR_u(S_i)$.
    It is easy to see that this is a best response in $BR_u(S_i)$, since a
    best response in $BR_u(S_i)$ is simply obtained by repeatedly adding
    effort to the edges with highest derivative. Moreover, $s'_u$
    matches the contributions on all edges to $A_i$, so once again we know
    that $\overline{BR}_u(S_i)$ is non-empty.
  \end{proof}


  Re-number the nodes $v_1,v_2,\ldots,v_n$ in the order that we wake
  them. We need to prove that the algorithm computes a pairwise
  equilibrium. By Lemma~\ref{lem:matchingContribution}, we know that
  all the contributions in the final solution are symmetric, and that
  node $v_i$ gets exactly the reward $\overline{BR}_{v_i}(S_i)$ in the
  final solution.

  To prove that the above algorithm computes a pairwise equilibrium,
  we show by induction on $i$ that node $v_i$ will never have
  incentive to deviate, either unilaterally or bilaterally. This is
  clearly true for $v_1$, since it obtains the maximum possible reward
  that it could have in {\em any} solution, which proves the base
  case. We now assume that this is true for all nodes earlier than
  $v_i$, and prove it for $v_i$ as well. It is clear that $v_i$ would
  not deviate unilaterally, since it is getting the reward of
  $BR_{v_i}(S_i)$. This is at least as good as any best response when
  it cannot control the strategies of any nodes except itself. By the
  inductive hypothesis, $v_i$ would not deviate bilaterally with a
  node $v_j$ such that $j<i$. $v_i$ would also not deviate bilaterally
  with a node $v_j$ such that $j>i$, since when forming
  $\overline{BR}_{v_i}(S_i)$ node $v_i$ can set the strategy of node
  $v_j$. So in $\overline{BR}_{v_i}(S_i)$ node $v_i$ achieves a reward
  better than any deviation possible with nodes from $S_i$. This
  completes the proof that our algorithm always finds a pairwise
  equilibrium.

  Now we will consider the case when all $h_e$ are {\em strictly}
  concave, and prove that there is a unique pairwise
  equilibrium. Consider the algorithm described above. It is greatly
  simplified for this case: since all $h_e$ are strictly concave, then
  $BR_v(S_i)$ consists of only a single strategy, and by
  Lemma~\ref{lem:matchingContribution}, this strategy is also in
  $\overline{BR}_v(S_i)$. We claim that when this algorithm assigns a
  strategy $s_v$ to a node $v$, then $v$ must have this strategy in
  {\em every} pairwise equilibrium. We will prove this by induction,
  so suppose this is true for all nodes earlier than $v=v_i$, but
  there is some pairwise equilibrium $s'$ where $v$ does not use the
  strategy $s_v\in
  \overline{BR}_v(S_i)$. 
  Since $s'_v\neq s_v$, then there must be some edge $e=(v,u)$ such
  that $s'_v(e)<s_v(e)$. If $u$ is a node considered earlier than $v$,
  then by the inductive hypothesis, we know that
  $s_u(e)=s_u'(e)$. $s_v$ is the unique best response of $v$ if it
  were able to control the strategies of nodes in $S_i$. This means
  that the gain that $v$ could obtain by moving some small amount of
  effort to edge $e$ is greater than the loss that it would obtain
  from removing effort from any edge to a node of $S_i$, and so $v$
  would have a unilateral deviation in $s'$. If instead $u\in S_i$,
  then the only way that it would not benefit $v$ to move some effort
  onto $e$ is if $s_u'(e)=s'_v(e)$. Since $v$ was chosen by the
  algorithm before $u$, we know that it would always benefit $u$ to
  move some effort onto edge $e$ in this case, since the derivative it
  would encounter there is higher than $u$ encounters on any other
  edge. Thus, there exists a bilateral deviation where both $u$ and
  $v$ move some effort onto edge $e$.
\end{proof}

For the case of strong equilibria, we observe that the arguments for
existence can be adapted, while the upper bounds of 2 on the price of
anarchy translate by restriction. In particular, consider a pairwise
equilibrium as described in the proof of
Theorem~\ref{thm:concaveMinExists}. Resilience to coalitional
deviations can be established in exactly the same way as above, i.e.,
the player from the coalition that was the first to be woken up has no
incentive to deviate.
\begin{cor}
  \label{cor:concaveMinStrong}
  A strong equilibrium always exists in games with $f_e(x,y) =
  h_e(\min(x,y))$ when all $h_e$ are continuous, piecewise
  differentiable, and concave. It is possible to compute strong
  equilibria efficiently within any desired precision. Moreover, if
  all $h_e$ are strictly concave, then this equilibrium is unique.
\end{cor}

\section{Maximum Effort Games}
\label{sec:max}

In this section we briefly consider $f_e(x,y) = h_e(\max(x,y))$ for
arbitrary monotonic increasing functions. Our results rely on the
following structural observation.

\begin{lemma}
  \label{lem:maxLemma}
  If there is a bilateral deviation that is strictly profitable for
  both players, then there is at least one player that has a
  profitable unilateral deviation.
\end{lemma}

\begin{proof}
  Suppose the bilateral deviation decreases the maximum effort on the
  joint edge. In this case, both players must receive more profit from
  other edges. This increase, however, can obviously also be realized
  by each player himself.

  Suppose the bilateral deviation increases the maximum effort on the
  joint edge. Then the player setting the maximum effort on the edge
  can obviously also do the corresponding strategy switch by himself,
  which yields the same outcome for him.
\end{proof}

It follows that a stable solution always exists, because the absence
of unilateral deviations implies that the state is also a pairwise
equilibrium. Furthermore, the total profit of all players is a
potential function of the game with respect to unilateral better
responses.\footnote{Note that the social welfare is not a potential
  function for bilateral deviations. Consider a path of length 3 with
  $h_e(x) = 2x$ on the outer edges and $h_e(x)=3x$ on the inner
  edge. The inner players have budget 1, the outer players budget
  0. If both inner players contribute to the outer edges, their
  utility is 2. If they both move all their effort to the inner edge,
  their utility becomes 3. Note, however, that the social welfare
  decreases from 8 to 6.} This implies that the social optimum is a
stable state and the price of stability is 1.

\begin{theorem}
  \label{thm:maxExist}
  A pairwise equilibrium always exists in games with $f_e(x,y) =
  h_e(\max(x,y))$ and arbitrary monotonic increasing functions $h_e$.
  The price of stability for pairwise equilibria is 1.
\end{theorem}

We can also easily derive a tight result on the price of anarchy for
arbitrary functions.

\begin{theorem}
  \label{thm:maximum}
  The price of anarchy for pairwise equilibria in network contribution
  games with $f_e(x,y) = h_e(\max(x,y))$ and arbitrary monotonic
  increasing functions $h_e$ is at most 2. This bound is tight for
  arbitrary convex functions.
\end{theorem}

\begin{proof}
  For an upper bound on the social welfare of the social optimum $s^*$
  consider each player $u$ and suppose that he optimizes his effort
  independently. This yields a reward $f_u$. Clearly, $w(s^*) \le
  2\sum_u f_u$. To see this, notice that in $s^*$, we can assume that
  every edge has contribution from only one direction. Let $E_u^*$ be
  the edges to which $u$ contributes in $s^*$. In this case, $u$'s
  reward from these edges is at most $f_u$. The reward of the other
  nodes because of these edges is also at most $f_u$.  Therefore, in
  total $w(s^*) \le 2\sum_u f_u$.

  On the other hand, in any pairwise equilibrium $s$ player $u$ will
  not accept less profit than $f_u$, because by a unilateral deviation
  he can always achieve (at least) the maxima used to optimize
  $f_u$. Thus, $w_u(s) \ge f_u$, and we have that
  \[
  w(s^*) \le 2\sum_{u \in V} f_u \le 2\sum_{u \in V} w_u(s) = 2
  w(s)~\enspace.
  \]

  Tightness follows from the following simple example. The graph is a
  path with four nodes $u_i$, for $i=1,...,4$. The interior players
  have budget 1, the leaf players have budget 0. The edges $e_1 =
  (u_1,u_2)$ and $e_2 = (u_2,u_3)$ have an arbitrary convex functions
  $h(x)$, the remaining edge $e_3$ has $h_{e_3}(x) = \varepsilon h(x)$,
  for an arbitrarily small $\varepsilon > 0$. A pairwise equilibrium $s$
  evolves when player $u_2$ spends his effort on $e_2$ and $u_3$ on
  $e_3$. This yields a total profit of $(2+2\varepsilon)h(1)$. The
  optimum evolves if $u_2$ contributes on $e_1$ and $u_3$ on $e_2$
  with total profit of $4h(1)$.
\end{proof}


\section{Approximate Equilibrium}
\label{sec:approx}

We showed above several classes of functions for which pairwise
equilibrium exists, and the price of anarchy is small. If we consider
{\em approximate} equilibria, however, the following theorem says that
this is always the case. By an $\alpha$-approximate equilibrium, we
will mean a solution where nodes may gain utility by deviating (either
unilaterally or bilaterally), but they will not gain more than a
factor of $\alpha$ utility because of this deviation.

\begin{theorem}\label{thm.approxEq}
  In network contribution games an optimum solution $s^*$ is a
  2-approximate equilibrium for any class of nonnegative reward
  functions.
\end{theorem}

\begin{proof}
  First, notice that $s^*$ is always stable against unilateral
  deviations. This is because when a node $v$ changes the effort it
  allocates to its adjacent edges unilaterally, then the only nodes
  affected are neighbors of $v$. If $C$ is the change in node $v$'s
  reward because of its unilateral deviation, then the total change in
  social welfare is exactly $2C$. Therefore, no node can improve their
  reward in $s^*$ using unilateral deviations.

  Now consider bilateral deviations, and assume for contradiction that
  nodes $u$ and $v$ have a bilateral deviation by adding some amounts
  $\delta_u$ and $\delta_v$ to edge $e=(u,v)$, which increases their
  rewards by more than a factor of 2. Let $z_u^* = w_u^* -
  f_e(s_u^*(e),s_v^*(e))$ and $z_v^* = w_v^* - f_e(s_u^*(e),s_v^*(e))$
  be the rewards of $u$ and $v$ in $s^*$ not counting edge $e$. We
  denote by $z_u$ and $z_v$ the same rewards after $u$ and $v$ deviate
  by adding effort to $e$, and therefore possibly taking effort away
  from other adjacent edges. In other words, the reward of $u$ before
  the deviation is $w_u^*(s^*) = z_u^*+f_e(s^*_u,s^*_v)$, and after
  the deviation it is $z_u+f_e(s^*_u+\delta_u,s^*_v+\delta_v)$. Note
  that this change cannot increase $w(s)$ over $w(s^*)$, therefore, we
  know that
  \begin{equation}\label{eqn.approxEq1}
    2z_u+2z_v+2f_e(s^*_u+\delta_u,s^*_v+\delta_v) \leq 2z_u^*+2z_v^*+2f_e(s^*_u,s^*_v)\enspace.
  \end{equation}
  On the other hand, since both $u$ and $v$ must improve their reward
  by more than a factor of 2, we know that
  \[ z_u+f_e(s^*_u+\delta_u,s^*_v+\delta_v) >
  2z_u^*+2f_e(s^*_u,s^*_v)\enspace,\]
  and
  \[ z_v+f_e(s^*_u+\delta_u,s^*_v+\delta_v) >
  2z_v^*+2f_e(s^*_u,s^*_v)\enspace.\]
  Adding the last two inequalities together, we obtain that
  \begin{eqnarray}
    z_u+z_v+2f_e(s^*_u+\delta_u,s^*_v+\delta_v) > 2z_u^*+2z_v^*+4f_e(s^*_u,s^*_v)\\
    \geq 2z_u+2z_v+2f_e(s^*_u+\delta_u,s^*_v+\delta_v)+2f_e(s^*_u,s^*_v)
  \end{eqnarray}
  which implies that $z_u+z_v+2f_e(s^*_u,s^*_v) < 0$, a contradiction.
\end{proof}

\section{Convergence}
\label{sec:convergence}

In this section we consider the convergence of round-based improvement
dynamics to pairwise equilibrium. Perhaps the most prominent variant
is best response, in which we deterministically and sequentially pick
one particular player or a pair of adjacent players in each round and
allow them to play a specific unilateral or bilateral deviation. While
convergence of such dynamics is desirable, a drawback is that
convergence could rely on the specific deterministic sequence of
deviations. Here we will consider less demanding processes that allow
players or pairs of players to be chosen at random to make deviations,
and we even allow concurrent deviations of more than one player or
pair.

We consider random best response, where we randomly pick either a
single player or one pair of adjacent players in each round and allow
them to play a unilateral or bilateral deviation. In each round, we
make this choice uniformly at random, i.e, a specific pair $(u,v)$ of
players gets the possibility to make a bilateral deviation with
probability $1/(n+m)$. In concurrent best response, each player
decides independently whether he wants to deviate unilaterally or
picks a neighbor for a bilateral deviation. Obviously, a bilateral
deviation can be played if and only if both players decide to pick
each other. Hence, in a given round a player $v$ decides to play a
unilateral deviation with probability $p_v = 1/(\text{deg}_v+1)$,
where deg$_v$ is the degree of $v$. A pair $(u,v)$ of players makes a
bilateral deviation with probability $p_u \cdot p_v$. Note that in
both dynamics, in expectation, after a polynomial number of rounds
each single player or pair of players gets the chance to play a
unilateral or bilateral deviation.

The name ``best response'' in our dynamics needs some more explanation
for bilateral deviations, because for a pair of players $(u,v)$ a
particular joint deviation $(s'_u,s'_v)$ might result in the best
reward for $u$ but not for $v$. In fact, there might be no joint
deviation that is simultaneously optimal for both players. In this
case the players should agree on one of the Pareto-optimal
alternatives.

In this section we consider special kinds of dynamics, which resolve
this issue in an intuitive way. The intuition is that if two players
decide to play a bilateral deviation, then these strategies should
also be unilateral best responses. We assume that players do not pick
bilateral deviations, in which they would change the strategies
unilaterally. More formally, we capture this intuition by the
following definition.
\begin{definition}
  A \emph{bilateral best response} for a pair $(u,v)$ of players in a
  state $s$ is a pair $(s'_u, s'_v)$ of strategies that is
  \begin{itemize}
  \item a profitable bilateral deviation, i.e.,
    $w_u(s'_u,s'_v,s_{-u,v}) > w_u(s)$, and $w_v(s'_u,s'_v,s_{-u,v}) >
    w_v(s)$, and
  \item a pair of mutual best responses, i.e.,
    $w_u(s'_u,s'_v,s_{-u,v}) \ge w_u(s''_u,s'_v,s_{-u,v})$ for every
    strategy $s''_u$ of player $u$, and similarly for $v$.
  \end{itemize}
\end{definition}

Note that, in principle, there might be states that allow a bilateral
deviation, but there exists no bilateral best response. The set of
states resilient to unilateral and bilateral best responses is a
superset of pairwise equilibria. Hence, it might not even be obvious
that dynamics using only bilateral best responses converge to pairwise
equilibria. Our results below, however, show that the latter is true
in many of the games for which we showed existence of pairwise
equilibria above.

\paragraph{General Convex Functions}
For games with strictly coordinate-convex functions, the concept of
bilateral best response reduces to a simple choice rule. In this case, a
unilateral best response of every player places the entire player budget
on a single edge. This implies that there is no bilateral best response
where players split their efforts. Thus, bilateral best responses come in
three different forms, in which the players allocate their efforts towards
their joint edge (1) both, (2) only one of them, or (3) none of them.

Consider two incident players $u$ and $v$ in state $s$ connected by edge
$e$. To compute a bilateral best response from one of the forms mentioned
above, we proceed in two phases. In the first phase, we try forms (2) and
(3) and remove all contributions from $e$. Then player $u$ independently
picks a unilateral best response under the assumption that $s_v(e) = 0$.
Note that in case of equal reward a player always prefers to put the
effort on $e$, because this might attract the other player to put effort
on $e$ as well and, by convexity, increase their own reward even further.
Similarly, we do this for player $v$. This yields a pair of ``virtual''
best responses under the condition that the other player does not
contribute to $e$. Now we have to check whether this is a bilateral best
response. In particular, if only one of the players puts effort on $e$, by
convexity it might become a unilateral best response for the other player
to put his effort on $e$ as well. If this is the case, the computed state
is not a pair of mutual best responses, thus the most profitable candidate
for a bilateral best response is of form (1).

If this is not the case, then by convexity one player is not willing
to contribute to $e$ at all. Hence, his virtual best response is a
unilateral best response even though the other player contributes to
$e$. For the other player, this means that the assumption made for the
virtual best response are satisfied, hence, we have found a set of
mutual best responses of the form (2) or (3). However, in this case,
this set might not be a bilateral best response because the players do
not improve over their current reward. We thus also check, whether a
state of form (1) is a better set of mutual best responses. Hence, we
consider the state of form (1) and the resulting reward for each
player. Each reward must be at least as large as that from the virtual
best responses, because otherwise the state does not represent a set
of mutual best responses. If this is the case, we accept $s'_u(e) =
B_u$ and $s'_v(e) = B_v$ as our candidate for the bilateral best
response. Otherwise, we use the pair of virtual best responses.

Note that our algorithm computes in each case the most profitable
candidate for a bilateral best response, and always finds a bilateral best
response if one exists. This can be verified directly for each of the
cases, in which there is a bilateral best response of forms (1), (2), or
(3).

\begin{theorem}
  \label{thm:convexConverge}
  Random and concurrent best response dynamics converge to a pairwise
  equilibrium in a polynomial number of rounds when all reward
  functions are strictly coordinate-convex and $f_e(0,x) = 0$ for all
  $e \in E$ and $x \ge 0$.
\end{theorem}

\begin{proof}
  Let us consider the edges in classes of their $c_{u,v}$ (c.f.\ proof
  of Theorem~\ref{thm.convex}). In particular, our analysis proceeds
  in phases. In phase 1, we restrict our attention to the first class
  of edges with the highest $c_{u,v}$ and the subgraph induced by
  these edges. Consider one such edge $e$ and suppose both players
  contribute their complete budgets to $e$. They are never again
  willing to participate in a bilateral deviation (not only a
  bilateral best response), because by strict convexity they achieve
  the maximum possible revenue. We will call such players
  \emph{stabilized}. Consider a first class edge $e = (u,v)$ where
  strategies $s_u(e) < B_u$ and $s_v(e) = B_v$. In this case, strict
  convexity, $f_e(x,y) = 0$ when $xy = 0$, and maximality of $c_{u,v}$
  imply that $s'_u(e) = B_u$ is a unilateral best response -
  independently of the current $s_u$. If both players have $s_u(e) <
  B_u$ and $s_v(e) < B_v$, then the same argument implies that both
  players allocating their full budget is a bilateral best response,
  again independently of what the current strategies of the players
  are. Note that each bilateral best response of two destabilized
  players enlarges the set of stabilized players. Phase 1 ends when
  there are no adjacent destabilized players with respect to first
  class edges, and this obviously takes only an expected number of
  time steps that is polynomial in $n$.

  After phase 1 has ended, we know that the stabilized players are
  never going to change their strategy again. Hence, for the purpose
  of our analysis, we drop the edges between stabilized players from
  consideration. The same can be done for all edges $e$ incident to
  exactly one stabilized player $v$, by artificially reducing $v$'s
  budget to $B_v = 0$ and noting that $f_e(B_u,0) = 0$. If there are
  remaining destabilized players, phase 2 begins, and we consider only
  the remaining players and the edges among them. In this graph, we
  again consider only the subgraph induced by edges with highest
  $c_{u,v}$. Again, we have the property that any pair of players
  contributing their full budget to such an edge is
  stabilized. Additionally, the same arguments show that for
  destabilized players there are always unilateral and/or bilateral
  best responses that result in investing the full budget,
  irrespective of the current strategy. Hence, after expected time
  polynomial in $n$, phase 2 ends and expands the set of stabilized
  players by at least 2.

  Repeated application of this argument shows that after expected time
  polynomial in $n$ either all players are stabilized or the remaining
  subgraph of destabilized players is empty. In this case, a pairwise
  equilibrium is reached. In particular, using unilateral and
  bilateral best responses suffices to stabilize all but an
  independent set of players. It is easy to observe that stabilized
  players have no profitable unilateral or bilateral
  deviations. Possibly remaining destabilized players in the end are
  only adjacent to stabilized players and therefore have no profitable
  unilateral or bilateral deviations. Thus, our dynamics converge to a
  pairwise equilibrium in expected polynomial time. This proves the
  theorem.
\end{proof}

\paragraph{Minimum Effort Games and Convex Functions}
In this section we show that there are games with infinite convergence
time of random and concurrent best response dynamics, although in each
step bilateral best responses are unique and can be found easily.

\begin{theorem}
  There are minimum effort games that have convex functions, uniform
  budgets, and starting states, from which any dynamics using only
  bilateral best responses does not converge to a stable state.
\end{theorem}

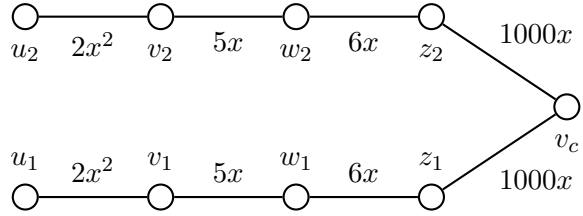
\begin{figure}
  \begin{center}
    \psset{unit=0.06}
    \begin{pspicture}(0,0)(130,50)

      \cnode(5,5){3}{A}
      \rput(5,13){$u_1$}
      \cnode(35,5){3}{B}
      \rput(35,13){$v_1$}
      \cnode(65,5){3}{C}
      \rput(65,13){$w_1$}
      \cnode(95,5){3}{D}
      \rput(95,13){$z_1$}

      \cnode(5,45){3}{E}
      \rput(5,37){$u_2$}
      \cnode(35,45){3}{F}
      \rput(35,37){$v_2$}
      \cnode(65,45){3}{G}
      \rput(65,37){$w_2$}
      \cnode(95,45){3}{H}
      \rput(95,37){$z_2$}

      \cnode(125,25){3}{I}
      \rput(125,17){$v_c$}

      \ncline{A}{B}
      \naput{$2x^2$}
      \ncline{B}{C}
      \naput{$5x$}
      \ncline{C}{D}
      \naput{$6x$}

      \ncline{D}{I}
      \nbput{$1000x$}
      \ncline{I}{H}
      \nbput{$1000x$}

      \ncline{H}{G}
      \naput{$6x$}
      \ncline{G}{F}
      \naput{$5x$}
      \ncline{F}{E}
      \naput{$2x^2$}

    \end{pspicture}
  \end{center}
  \caption{\label{fig:noConverge} A minimum effort game with convex
    reward functions and uniform budgets $B_v = 2$. There is a
    starting state such that no sequence of bilateral best responses
    can reach a stable state.}

\end{figure}

\begin{proof}
  We consider two paths of length 4 as in the games of
  Example~\ref{exm:minNoEq} and introduce a new player $v_c$ as shown
  in Figure~\ref{fig:noConverge}. All players have budget 2. In our
  starting state $s$ we assign all incident players to contribute 1 to
  the edges $(z_1,v_c)$ and $(z_2,v_c)$. This yields a maximum revenue
  of 2000 for $v_c$. As long as this remains the case, $v_c$ will
  never participate in a bilateral deviation. In turn, in every
  unilateral best response players $z_1$ and $z_2$ will match the
  contribution of $v_c$ towards their joint edges. Note that this
  essentially creates the budget restriction for the $z$-players that
  is necessary to show non-existence of a pairwise equilibrium in
  Example~\ref{exm:minNoEq}. It remains to show that we can implement
  the cycling of dynamics in terms of bilateral best responses. For
  this, note that for player $u_1$ it is always a unilateral best
  response to match any contribution of $v_1$ on their joint edge
  (similarly for $u_2$ and $v_2$). The same is true for $z_1$, he will
  match the contribution of $w_1$ up to an effort of 1. Finally, the
  joint deviations of players $v_1$ and $w_1$ are bilateral best
  responses as well. This implies that the cycling dynamics outlined
  above in Example~\ref{exm:minNoEq} remain present when we restrict
  to bilateral best responses. Thus, no stable state can be reached.
\end{proof}

Observe that in this game there are sequences of bilateral deviations that
converge to a pairwise equilibrium, but the bilateral deviations are not
bilateral best responses. Consider an arbitrary cycling sequence of
bilateral deviations from our starting state, and w.l.o.g.\ consider the
cycling dynamics happen on the upper path in Figure~\ref{fig:noConverge}.
Then at some point we will see a bilateral deviation of $w_2$ and $z_2$,
in which on their joint edge $z_2$ contributes 1 and $w_2$ increases his
effort. This creates a strict improvement of utility for both of them.
Note that a bilateral deviation allows both $w_2$ and $z_2$ to change
their strategies in arbitrary manner. Thus, while increasing his
contribution towards $w_2$, $z_2$ can also simultaneously decrease his
contribution towards $v_c$. If the decrease is very tiny, the increase in
reward on the edge to $w_2$ outweighs the decrease of reward on the edge
to $v_c$. In this way, this deviation still generates a strict improvement
of utility for $z_2$. Hence, both $z_2$ and $w_2$ would make strict
improvement although $z_2$ decreases the contribution towards $v_c$ by a
tiny amount, hence this represents a profitable bilateral deviation (but
obviously not a best response). Afterwards, the balance for $v_c$ is
broken, and $v_c$ and $z_1$ have a bilateral deviation to put all effort
on their joint edge. This quickly leads to a pairwise equilibrium.
Naturally, the argument works symmetrically for $z_1$. However, such an
evolution is quite unreasonable, as it is always in the interest of the
$z$-players to keep their contribution towards $v_c$ as high as possible.

\paragraph{Minimum Effort Games and Concave Functions}
For concave functions, we can use the following simple rule to find a
bilateral best response. Consider two incident players $u$ and $v$ in
state $s$ connected by edge $e$. In the first phase we consider each
player independently and compute a unilateral best response $s'_u$ and
$s'_v$ under the assumption that the other player would match his
contribution on $e$. Then we fix the strategy of the player $u$ for
which $s'_u(e) < s'_v(e)$. In the state $s'$, player $u$ is perfectly
happy with his choice and would not participate in any bilateral
deviation. However, player $v$ might be willing to deviate, so we
recalculate a unilateral best response for $v$ under the condition
that $s'_v(e) \le s'_u(e)$. Note that, due to concavity of the
functions, $v$ has a unilateral best response that matches
$s'_u(e)$. This yields a pair of mutual best responses: $u$ has best
possible utility (even if it were able to control $v$'s strategy), and
$v$ has the best possible utility given $u$'s strategy. As usual, the
players switch to $(s'_u, s'_v)$ if and only if it is a profitable
bilateral deviation.

\begin{theorem}
  \label{thm:concaveConverge}
  Random and concurrent best response dynamics converge to a pairwise
  equilibrium when all reward functions are $f_e(x,y) =
  h_e(\min(x,y))$ with differentiable and strictly concave $h_e$.
\end{theorem}

\begin{proof}
  We measure progress in terms of the derivatives of the edges. For a
  state $s$ consider an edge with highest derivative $e_{max} =
  \arg\max_{e \in E} h'_e$, where $h'_{(u,v)} =
  h_e'(\min(s_u(e),s_v(e)))$. Obviously, $h'_{max} = h'_{e_{max}} \ge
  h'_e$ for any other edge $e \in E$, so in any unilateral or
  bilateral best response the incident players will not try to remove
  effort from this edge once $s_u(e)=s_v(e)$. Edge $e_{max} =(u,v)$ is
  \emph{stabilized} if there is a player $u$ with $h'_{max} = h'_{e'}$
  for every edge $e' = (u,v') \in E$ and spending all his budget,
  i.e., $\sum_{e'=(u,v') \in E} s_{e'}(u) = B_u$. In this case, no
  player will remove effort from $e_{max}$, but at least one player
  has no interest in increasing effort on $e_{max}$.

  We now consider the dynamics starting in a state $s$ and the set of
  non-stabilized edges $E_{max}$ with maximum derivative $h'_{max}$
  among non-stabilized edges. Suppose that a bilateral deviation
  results in a reduction of the minimum effort on any edge $e$ to a
  value $x$ with $h'_e(x) > h'_{max}$. This is a contradiction to
  $h'_{max}$ being the currently highest derivative value and the
  deviation being composed of mutual unilateral best responses. Hence,
  the value $h'_{max}$ will never increase over the run of the
  dynamics. As an edge $e \in E_{max}$ with highest derivative is not
  stabilized, both incident players have other incident edges with
  strictly smaller derivative. Hence, if they play a bilateral best
  response, they strictly increase effort on $e$ while strictly
  decreasing effort on other edges. By strict concavity this implies
  that after the step $h'_e < h'_{max}$. In addition, both players
  picking a best response means that the derivative of all edges that
  were previously lower than $h'_{max}$ now remain at most
  $h'_e$. This means that no new edge with derivative value $h'_{max}$
  is created, but $e$ is removed. Thus, in each such step we either
  increase the number of stabilized edges, or we decrease the number
  of edges of highest derivative among non-stabilized edges. As such a
  step is played after a finite number of steps in expectation, this
  argument proves convergence.

  It remains to show that the resulting state, in which all edges are
  stabilized, is resilient to all unilateral and bilateral deviations
  and not only against the type of bilateral best responses we used to
  converge to it. Here we can apply an inductive argument similar to
  Theorem~\ref{thm:concaveMinExists} that no profitable bilateral
  deviation exists and the state is indeed a pairwise
  equilibrium. Note that the argument simplifies quite drastically for
  the case of \emph{strictly} concave functions. In particular, we
  consider the edge with maximum derivative. For at least one incident
  player $v$, all edges with positive minimum effort have the same
  derivative, hence this player will never change his strategy. In
  addition, the other adjacent players have an incentive to keep their
  efforts on the edges with $v$. Thus, we can remove $v$, reduce the
  budgets of incident players and iterate. This proves the theorem.
\end{proof}

The previous proof shows that convergence is achieved in the limit, but
the decrease of the maximum derivative value is not bounded. If the state
is close to a pairwise equilibrium, the changes could become arbitrarily
small, and the convergence time until reaching the exact equilibrium could
well be infinite.

\section{General Contribution Games}
\label{sec:general}

In this section we generalize some of our results to general
contribution games. However, a detailed study of such general games
remains as an open problem. A general contribution game can be
represented by a hypergraph $G=(V,E)$. The set of nodes $V$ is the set
of players, and each edge $e \in E$ is a hyperedge $e \subseteq 2^V$
and represents a joint project of a subset of players. Reward
functions and player utilities are defined as before. In particular,
using the notation $s_e = (s_u)_{u \in e}$ we get reward functions
$f_e(s_e)$ with $f_e : (\R_{\ge 0})^{|e|} \to \R_{\ge 0}$. In this
case, we extend our stability concept to \emph{setwise equilibrium}
that is resilient against all deviations of all player sets that are a
subset of any hyperedge. In a setwise equilibrium no (sub-)set of
players incident to the same hyperedge has an improving move, i.e., a
combination of strategies for the players such that every player of
the subset strictly improves. More formally, a \emph{setwise
  equilibrium} $s$ is a state such that for every edge $e \in E$ and
every player subset $U \subseteq e$ we have that for every possible
deviation $s'_U = (s'_u)_{u\in U}$ there is at least one player $v \in
U$ that does not strictly improve $w_v(s) \ge w_v(s'_U,s_{-U})$, where
$s_{-U} = (s_u)_{u \in V-U}$. Note that this definition includes all
unilateral deviations as a special case. The most central parameter in
this context will be the size of the largest project $k = \max_{e \in
  E} |e|$.

We note in passing that Nash equilibria without resilience to
multilateral improving moves are again always guaranteed for all
reward functions. It is easy to observe that $\Phi(s) = \sum_{e \in E}
f_e(s_e)$ is an exact potential function for the game. Note that this
is equal to $k\cdot w(s)$ and thus equivalent to $w(s)$ only for
uniform hypergraphs, in which all hyperedges have the same cardinality
$k$. In these cases the price of stability for Nash equilibria is
always 1. Otherwise, it is easy to construct simple examples, in which
all Nash equilibria (and therefore all setwise equilibria) must be
suboptimal.\footnote{Consider a maximum effort game with players $u$,
  $v_1$, $v_2$, and $v_3$ and edges $e = \{u,v_1\}$ and $g =
  \{u,v_1,v_2\}$. Budgets $B_u = 1$ and $B_{v_i} = 0$ for $i =
  1,2,3$. Rewards are given by a convex function $h(x)$ for edge $g$
  and a function $(1+\varepsilon)h(x)$ for $e$. Note that in any Nash
  equilibrium $s_u(e) = 1$. For small $\varepsilon < 1/2$ this gives
  $w(s) = (2+2\varepsilon)h(1)$, whereas we have higher welfare $w(s) =
  3h(1)$ when $s_u(g) = 1$.}

\paragraph{Convex Functions}
For general convex functions we extend the functions of class
$\mathcal{C}$ to multiple dimensions. In particular, the functions are
coordinate-convex and have non-negative mixed partial second
derivatives for any pair of dimensions. We first observe that the
proof of Theorem~\ref{thm.convex} can be adjusted easily to general
games if functions $f_e$ are coordinate-convex and $f_e(s_e) = 0$
whenever $s_u(e) = 0$ for at least one $u \in e$.

\begin{cor}
  \label{cor.convex}
  A setwise equilibrium always exists and can be computed efficiently
  when all reward functions $f_e$ are coordinate-convex and
  $f_e(s_e)=0$ whenever $s_u(e) = 0$ for at least one $u \in e$.
\end{cor}

Note that we can also adjust the proof of
Claim~\ref{claim.optIntegral} for optimum solutions in a
straightforward way. In particular, to obtain the social welfare for
projects $e$, $e_1$ and $e_2$ we simply multiply each occurrence of
the functions $f$, $f_1$ and $f_2$ in the formulas by the
corresponding cardinalities of their edges. This does not change the
reasoning and proves an analogous statement of
Claim~\ref{claim.optIntegral} also for general games. The actual proof
of Theorem~\ref{thm.classC} then is a simple accounting argument that
relies on the cardinality of the projects. The observation that the
difference between $\sum_e w_e(s^*)$ and $w(s^*)$ is bounded by $k$
yields the following corollary. As previously, these results directly
extend to strong equilibria, as well.
\begin{cor}
  \label{cor.classC}
  For the class of general contribution games with reward functions in
  class $\mathcal{C}$ for all $e \in E$ that have a setwise
  equilibrium, the price of anarchy for setwise equilibria is at most
  $k$.
\end{cor}
%
%
%

\paragraph{Minimum Effort Games}
For minimum effort games some of our arguments translate directly to
the treatment of general games. For existence with convex functions
and uniform budgets, we can apply the same ``greedy matching''
argument and wake up players until every hyperedge is incident to at
least one awake player. The argument that this creates a setwise
equilibrium is almost identical to the one given in
Theorem~\ref{thm:convexMinExists} for pairwise equilibria. This yields
the following corollary.
\begin{cor}
  \label{cor:convexMinExists}
  A setwise equilibrium always exists in games with uniform budgets
  and $f_e(s_e) = h_e(\min_{u \in e} s_u(e))$ when all $h_e$ are
  convex. If all $h_e$ are strictly convex, all setwise equilibria are
  integral.
\end{cor}
The duality analysis for the price of anarchy in
Theorem~\ref{thm:PoAconvexMinUniform} can be carried out as well. In
this case, however, the crucial inequality reads $s_e^* \cdot \sum_{u
  \in e} y_u \ge s_e^* f_e(1) \ge f_e(s_e^*)$. This results in a price
of anarchy of $k$.
\begin{cor}
  \label{cor.PoAconvexMinUniform}
  The price of anarchy for setwise equilibria in general contribution
  games is at most $k$ when all reward functions $f_e(s_e) =
  h_e(\min_{u \in e} s_u(e))$ with convex $h_e$, and budgets are
  uniform.
\end{cor}
Again, both corollaries extend also to strong equilibria.
%
%

\paragraph{Maximum Effort Games}
For maximum effort games it is not possible to extend the main insight
in Lemma~\ref{lem:maxLemma} to general games. There are general
maximum effort games without setwise equilibria. This holds even for
pairwise equilibria in network contribution games, in which the graph
$G$ is not simple, i.e., if we allow multiple edges between agents.
\begin{example}\rm
  We consider a simple game that in essence implements a Prisoner's
  Dilemma. There is a path of four players $u$, $v$, $w$ and $z$, with
  edges $e_1 = (u,v)$, $e_2 = (v,w)$ and $e_3 = (w,z)$. In addition,
  there is a second edge $e_4 = (v,w)$ between $v$ and $w$. The
  budgets are $B_u = B_z = 0$ and $B_v = B_w = 1$. The reward
  functions are $h_{e_1}(x) = h_{e_3}(x) = 3x$, $h_{e_2}(x) =
  h_{e_4}(x) = 2x$. Note that for $v$ and $w$ it is a unilateral
  dominant strategy to put all effort on edges $e_1$ and $e_3$,
  respectively. However, in that case $v$ and $w$ can jointly increase
  their reward by allocating all effort to $e_2$ and $e_4$,
  respectively.
\end{example}
For general maximum effort games characterizing the existence and
computational complexity of pairwise, setwise, and strong equilibria
is an interesting open problem.

\section{Conclusions and Open Problems}
\label{sec:conclude}

In this paper we have proposed and studied a simple model of
contribution games, in which agents can invest a fixed budget into
different relationships. Our results show that collaboration between
pairs of players can lead to instabilities and non-existence of
pairwise equilibria. For certain classes of functions, the existence
of pairwise equilibria is even \classNP-hard to decide. This implies
that it is impossible to decide efficiently if a set of players in a
game can reach a pairwise equilibrium. For many interesting classes of
games, however, we are able to show existence and bound the price of
anarchy to 2. This includes, for instance, a class of games with
general convex functions, or minimum effort games with concave
functions. Here we are also able to show that best response dynamics
converge to pairwise equilibria.

There is a large variety of open problems that stem from our work. The
obvious open problem is to adjust our results for the network case to
general set systems and general contribution games. While some of our
proofs can be extended in a straightforward way, many open problems,
most prominently for concave functions, remain.

Another obvious direction is to identify other relevant classes of
games within our model and prove existence and tight bounds on the
price of anarchy. Another interesting aspect is, for instance, the
effect of capacity constraints, i.e., restrictions on the effort that
a player can invest into a particular project.

More generally, instead of a total budget a player might have a
function that characterizes how much he has to ``pay'' for the total
effort that he invests in all projects. Such ``price'' functions are
often assumed to be linear or convex (e.g.,
in~\cite{BramoulleJET07,Huyck90}).

Finally, an intriguing adjustment that we outlined in the introduction
is to view the projects as instances of the combinatorial agency
framework and to examine equilibria in this more extended model.

\section*{Acknowledgements}
The authors would like to thank Ramamohan Paturi for interesting
discussions about the model.

\bibliographystyle{plain}

\begin{thebibliography}{10}

\bibitem{Ackermann11}
Heiner Ackermann, Paul Goldberg, Vahab Mirrokni, Heiko R{\"o}glin, and Berthold
  V{\"o}cking.
\newblock Uncoordinated two-sided matching markets.
\newblock {\em SIAM J. Comput.}, 40(1):92--106, 2011.

\bibitem{Albers06}
Susanne Albers, Stefan Eilts, Eyal Even-Dar, Yishay Mansour, and Liam Roditty.
\newblock On {N}ash equilibria for a network creation game.
\newblock In {\em Proc.\ 17th Symp.\ Discrete Algorithms (SODA)}, pages 89--98,
  2006.

\bibitem{AlosFerrerMini10}
Carlos Al\'{o}s-Ferrer and Simon Weidenholzer.
\newblock Imitation in minimum effort network games.
\newblock Unpublished manuscript, 2010.

\bibitem{Andelman09}
Nir Andelman, Michal Feldman, and Yishay Mansour.
\newblock Strong price of anarchy.
\newblock {\em Games Econom.\ Behav.}, 65(2):289--317, 2009.

\bibitem{Anderson01}
Simon Anderson, Jacob Goeree, and Charles Holt.
\newblock Minimum-effort coordination games: {S}tochastic potential and logit
  equilibrium.
\newblock {\em Games Econom.\ Behav.}, 34(2):177--199, 2001.

\bibitem{AnshelevichD09}
Elliot Anshelevich, Sanmay Das, and Yonatan Naamad.
\newblock Anarchy, stability, and utopia: {C}reating better matchings.
\newblock In {\em Proc.\ 2nd Intl.\ Symp.\ Algorithmic Game Theory (SAGT)},
  pages 159--170, 2009.

\bibitem{Babaioff06}
Moshe Babaioff, Michal Feldman, and Noam Nisan.
\newblock Combinatorial agency.
\newblock In {\em Proc.\ 7th Conf.\ Electronic Commerce (EC)}, page~28, 2006.

\bibitem{BabaioffSAGT09}
Moshe Babaioff, Michal Feldman, and Noam Nisan.
\newblock Free-riding and free-labor in combinatorial agency.
\newblock In {\em Proc.\ 2nd Intl.\ Symp.\ Algorithmic Game Theory (SAGT)},
  pages 109--121, 2009.

\bibitem{Ballester06}
Coralio Ballester, Antoni Calv{\'o}-Armengol, and Yves Zenou.
\newblock Who's who in networks. {W}anted: {T}he key player.
\newblock {\em Econometrica}, 74(5):1403--1417, 2006.

\bibitem{Bei09}
Xiaohui Bei, Wei Chen, Shang-Hua Teng, Jialin Zhang, and Jiajie Zhu.
\newblock Bounded budget betweenness centrality game for strategic network
  formations.
\newblock In {\em Proc.\ 17th European Symposium on Algorithms (ESA)}, pages
  227--238, 2009.

\bibitem{Bhaskar09}
Umang Bhaskar, Lisa Fleischer, Darrell Hoy, and Chien-Chung Huang.
\newblock Equilibria of atomic flow games are not unique.
\newblock In {\em Proc.\ 20th Symp.\ Discrete Algorithms (SODA)}, pages
  748--757, 2009.

\bibitem{Bloch09}
Francis Bloch and Bhaskar Dutta.
\newblock Communication networks with endogenous link strength.
\newblock {\em Games Econom.\ Behav.}, 66(1):39--56, 2009.

\bibitem{Bornstein02}
Gary Bornstein, Uri Gneezy, and Rosmarie Nagel.
\newblock The effect of intergroup competition on group coordination: {A}n
  experimental study.
\newblock {\em Games Econom.\ Behav.}, 41(1):1--25, 2002.

\bibitem{BramoulleJET07}
Yann Bramoull{\'e} and Rachel Kranton.
\newblock Public goods in networks.
\newblock {\em J. Econ.\ Theory}, 135(1):478--494, 2007.

\bibitem{Brandes08}
Ulrik Brandes, Martin Hoefer, and Bobo Nick.
\newblock Network creation games with disconnected equilibria.
\newblock In {\em Proc.\ 4th Intl.\ Workshop Internet \& Network Economics
  (WINE)}, pages 394--401, 2008.

\bibitem{ChakrabortyEC09}
Tanmoy Chakraborty, Michael Kearns, and Sanjeev Khanna.
\newblock Network bargaining: {A}lgorithms and structural results.
\newblock In {\em Proc.\ 10th Conf.\ Electronic Commerce (EC)}, pages 159--168,
  2009.

\bibitem{Chaudhuri08}
Ananish Chaudhuri, Andrew Schotter, and Barry Sopher.
\newblock Talking ourselves to efficiency: {C}oordination in inter-generational
  minimum effort games with private, almost common and common knowledge of
  advice.
\newblock {\em The Economic Journal}, 119(534):91--122, 2008.

\bibitem{Cominetti09}
Roberto Cominetti, Jos{\'e} Correa, and Nicolas~Stier Moses.
\newblock The impact of oligopolistic competition in networks.
\newblock {\em Oper.\ Res.}, 57(6):1421--1437, 2009.

\bibitem{Conitzer04}
Vincent Conitzer and Tuomas Sandholm.
\newblock Expressive negotiation over donations to charities.
\newblock In {\em Proc.\ 5th Conf.\ Electronic Commerce (EC)}, pages 51--60,
  2004.

\bibitem{Corbo07}
Jacomo Corbo, Antoni Calv{\'o}-Armengol, and David Parkes.
\newblock The importance of network topology in local contribution games.
\newblock In {\em Proc.\ 3rd Intl.\ Workshop Internet \& Network Economics
  (WINE)}, pages 388--395, 2007.

\bibitem{DaskalakisMinMax09}
Constantinos Daskalakis and Christos Papadimitriou.
\newblock On a network generalization of the {M}inmax theorem.
\newblock In {\em Proc.\ 36th Intl.\ Coll.\ Automata, Languages and Programming
  (ICALP)}, pages 423--434, 2009.

\bibitem{Davis09}
Joshua Davis, Zachary Goldman, Jacob Hilty, Elizabeth Koch, David Liben-Nowell,
  Alexa Sharp, Tom Wexler, and Emma Zhou.
\newblock Equilibria and efficiency loss in games on networks.
\newblock In {\em Proc.\ 2009 Intl.\ Conf.\ Computational Science and
  Engineering}, volume~4, pages 82--89, 2009.

\bibitem{Devetag07}
Giovanna Devetag and Andreas Ortmann.
\newblock When and why? {A} critical survey on coordination failure in the
  laboratory.
\newblock {\em Experimental Economics}, 10(3):331--344, 2007.

\bibitem{Dufwenberg05}
Martin Dufwenberg and Uri Gneezy.
\newblock Gender and coordination.
\newblock In A.~Rapoport and R.~Zwick, editors, {\em Experimental Business
  Research}, volume~3, pages 253--262. Kluwer, 2005.

\bibitem{Fabrikant03}
Alex Fabrikant, Ankur Luthera, Elitza Maneva, Christos Papadimitriou, and Scott
  Shenker.
\newblock On a network creation game.
\newblock In {\em Proc.\ 22nd Symp.\ Principles of Distributed Computing
  (PODC)}, pages 347--351, 2003.

\bibitem{Fatas06}
Enrique Fatas, Tibor Neugebauer, and Javier Perote.
\newblock Within-team competition in the minimum effort coordination game.
\newblock {\em Pacific Econ. Rev.}, 11(2):247--266, 2006.

\bibitem{Ghosh08}
Arpita Ghosh and Mohammad Mahdian.
\newblock Charity auctions on social networks.
\newblock In {\em Proc.\ 19th Symp.\ Discrete Algorithms (SODA)}, pages
  1019--1028, 2008.

\bibitem{Gusfield89}
Dan Gusfield and Robert Irving.
\newblock {\em The Stable Marriage Problem: {S}tructure and Algorithms}.
\newblock MIT Press, 1989.

\bibitem{HoeferDISC09}
Martin Hoefer and Siddharth Suri.
\newblock Dynamics in network interaction games.
\newblock In {\em Proc.\ 23rd Intl.\ Symp.\ Distributed Computing (DISC)},
  pages 294--308, 2009.

\bibitem{Jackson08}
Matthew Jackson.
\newblock {\em Social and Economic Networks}.
\newblock Princeton University Press, 2008.

\bibitem{Jackson96}
Matthew Jackson and Asher Wolinsky.
\newblock A strategic model of social and economic networks.
\newblock {\em J. Econ.\ Theory}, 71(1):44--74, 1996.

\bibitem{Kleinberg08}
Jon Kleinberg, Siddharth Suri, {\'E}va Tardos, and Tom Wexler.
\newblock Strategic network formation with structural holes.
\newblock In {\em Proc.\ 9th Conf.\ Electronic Commerce (EC)}, pages 284--293,
  2008.

\bibitem{KleinbergSTOC08}
Jon Kleinberg and {\'E}va Tardos.
\newblock Balanced outcomes in social exchange networks.
\newblock In {\em Proc.\ 40th Symp.\ Theory of Computing (STOC)}, pages
  295--304, 2008.

\bibitem{Laoutaris08}
Nikolaos Laoutaris, Laura Poplawski, Rajmohan Rajaraman, Ravi Sundaram, and
  Shang-Hua Teng.
\newblock Bounded budget connection {(BBC)} games or how to make friends and
  influence people, on a budget.
\newblock In {\em Proc.\ 27th Symp.\ Principles of Distributed Computing
  (PODC)}, pages 165--174, 2008.

\bibitem{Marden08}
Jason Marden and Adam Wierman.
\newblock Distributed welfare games.
\newblock In {\em Proc.\ 47th IEEE Conf.\ Decision and Control}, 2008.

\bibitem{Nesterov94}
Yurii Nesterov and Arkadii Nemirovskii.
\newblock {\em Interior Point Polynomial Algorithms in Convex Programming},
  volume~13 of {\em SIAM Studies in Applied Mathematics}.
\newblock SIAM, 1994.

\bibitem{Orda93}
Ariel Orda, Raphael Rom, and Nahum Shimkin.
\newblock Competitive routing in multiuser communication networks.
\newblock {\em IEEE/ACM Trans.\ Netw.}, 1(5):510--521, 1993.

\bibitem{Riechmann08}
Thomas Riechmann and Joachim Weimann.
\newblock Competition as a coordination device: {E}xperimental evidence from a
  minimum effort coordination game.
\newblock {\em European J.\ Political Economy}, 24(2):437--454, 2008.

\bibitem{Huyck90}
John van Huyck, Raymond Battalio, and Richard Beil.
\newblock Tacit coordination games, strategic uncertainty and coordination
  failure.
\newblock {\em Amer.\ Econ.\ Rev.}, 80(1):234--248, 1990.

\end{thebibliography}

\end{document}